\documentclass[letterpaper, 10 pt, conference]{ieeeconf}  

\IEEEoverridecommandlockouts                              

\usepackage{amsfonts}
\usepackage{amstext}
\usepackage{mathtools}
\usepackage{cases}
\usepackage{enumerate}
\usepackage{subfigure}
\usepackage{amssymb}
\usepackage{float}
\usepackage{comment}
\usepackage[dvipsnames]{xcolor}
\usepackage{hyperref}

\newtheorem{thm}{\textbf{\text{Theorem}}}
\newtheorem{lem}{\textbf{\text{Lemma}}}
\newtheorem{pro}{\textbf{\text{Proposition}}}

\newtheorem{rmk}{\textbf{\text{Remark}}}
\newtheorem{cnd}{\textbf{\text{Condition}}}

\newcommand{\ie}{\textit{i.e.}}
\newcommand{\eg}{\textit{e.g.}}

\title{\LARGE \bf
Global Attitude Synchronization for Heterogeneous Multi-agent Systems on $SO(3)$}
\author{Mouaad Boughellaba, Soulaimane Berkane, and Abdelhamid Tayebi
\thanks{This work was supported by the National Sciences and Engineering Research Council of Canada (NSERC), under the grants NSERC-DG RGPIN 2020-06270 and NSERC-DG RGPIN-2020-04759, and by Fonds de recherche du Qu\'ebec (FRQ). Preliminary results of this paper have been presented in \cite{Mouaad_ACC24}.} 
\thanks{M. Boughellaba and A. Tayebi are with the Department of Electrical Engineering, Lakehead University, Thunder Bay, ON P7B 5E1, Canada \tt\small \{mboughel,atayebi\}@lakeheadu.ca.} 
\thanks{S. Berkane is with the Department of Computer Science and Engineering, University of Quebec in Outaouais, Gatineau, QC, Canada. {\tt\small Soulaimane.Berkane@uqo.ca}}
}

\begin{document}

\maketitle
\thispagestyle{empty}
\pagestyle{empty}

\begin{abstract}
In this paper, we address the problem of attitude synchronization for a group of rigid body systems evolving on $SO(3)$. The interaction among these systems is modeled through an undirected, connected, and acyclic graph topology. First, we present an almost global continuous distributed attitude synchronization scheme with rigorously proven stability guarantees. Thereafter, we propose two global distributed hybrid attitude synchronization schemes on $SO(3)$. The first scheme is a hybrid control law that leverages angular velocities and relative orientations to achieve global alignment to a common orientation. The second scheme eliminates the dependence on angular velocities by introducing dynamic auxiliary variables, while ensuring global asymptotic attitude synchronization. This velocity-free control scheme relies exclusively on attitude information. The proposed schemes are applicable to heterogeneous multi-agent systems, where agents may have distinct inertia matrices. Simulation results are provided to illustrate the effectiveness of the proposed distributed attitude synchronization schemes.
\end{abstract}

\section{Introduction}
\label{sec:introduction}
Attitude synchronization for multi-agent rigid-body systems consists of aligning the agent's orientations to a common orientation using local information exchange. This problem has garnered considerable attention from the research community over the past few decades due to its significant implications in various areas. For instance, many of the existing multi-agent rigid body formation control schemes assume that the agents' absolute orientations are known to allow the use of local relative measurements (\eg, positions, distances, or bearings) in the formation control laws. However, if the agents' absolute attitudes are unknown, they can still achieve the desired formation up to a common rotation by first synchronizing their attitudes to a common orientation and then using this common orientation together with local relative measurements in the formation control law. Note that the two tasks (\ie, the attitude synchronization and formation control) can be performed simultaneously \cite{Oh_TAC2014,Moshtagh_TAC2009}. \\
\indent A number of works have investigated the problem of attitude synchronization using different attitude parameterizations such as the Euler Angles (EA), Modified Rodriguez Parameters (MRP), and unit quaternions. The authors in \cite{Dimos_SCL_2009,Bayezit_TIE_2013,Ren_TCAT_2010,Xin_automatic_2020,Ziyang_automatic_2010,Chen_TAES2019} used EA and MRP representations to study the attitude synchronization problems. Unfortunately, these attitude representations evolve on the Euclidean space $\mathbb{R}^3$ and achieve only local synchronization due to the fact that the EA and MPR are not homomorphic to $SO(3)$ \cite{Ren_TCAT_2010}. Since the unit-quaternion represents the attitude of a rigid body globally \cite{Shuster1993ASO}, several works have addressed the attitude synchronization problem using this representation \cite{Ren_IJACSP2007,BAI20083170,Liu_automatic2018,Pedro_TAC2020,SAVINO2020142,Zhang_TCS2022}. In the same context, the authors in \cite{Abdessameud_TAC2009,Abdessameud_TAC2012} proposed quaternion-based attitude synchronization schemes that use the virtual dynamics approach, initially proposed in \cite{Tayebi_TAC2008}, to eliminate the need for the angular velocity measurements. Although the unit-quaternion representation does not suffer from the singularity problem, unit-quaternion space is a double cover of the special orthogonal group $SO(3)$. Consequently, the use of unit-quaternion, without extra care, can result in the undesirable unwinding phenomenon. Motivated by this, the authors in \cite{Mayhew_TAC2011,GUI2018225,Huang_automatic2021} proposed hybrid quaternion-based attitude synchronization schemes endowed with global asymptotic stability guarantees, while  effectively avoiding the unwinding phenomenon through the use of an appropriately designed logic variable to determine the sign of the torque input.\\
\indent Unlike other attitude parameterizations, the rotation matrix representation, which belongs to the special orthogonal group $SO(3)$, is the only representation that uniquely and globally represents the attitude of a rigid body. However, the topological obstruction to global asymptotic stability induced by the fact that $SO(3)$ is not homeomorphic to any Euclidean space \cite{Koditschek,Bhat_SCL2000}, poses a challenge in extending the classical Euclidean consensus schemes to consensus schemes on smooth manifolds such as $SO(3)$. Despite this challenge, several attitude synchronization schemes on $SO(3)$ have been proposed in the literature (\eg, \cite{Maadani_ACC2020,Tran_ACC22,Tron_CDC2012,Tron_TAC2012,Markdahl_TAC2020,SARLETTE2009572,SARLETTE20072232,Alain_SIAM,Wei_TAC2018,Maadani_2022}). Unfortunately, none of these papers was able to provide global asymptotic stability results.
 Therefore, the design of attitude synchronization schemes (with and without angular velocity measurements) on $SO(3)$ with \textit{global asymptotic stability} guarantees remains an open problem.
A promising way to circumvent this fundamental limitation is through the use of hybrid control techniques. Indeed, while continuous time-invariant feedback cannot overcome the topological obstruction on $SO(3)$, hybrid dynamics allow one to combine continuous flows with discrete switching mechanisms in order to eliminate undesired equilibria and achieve global attractivity (see, \eg, \cite{Mayhew_TAC2013,Lee2015Global, Berkane_TAC2016,miaomiao_TAC2022}). Hybrid approaches have proven successful in stabilization and tracking problems on $SO(3)$, yet their systematic integration into the multi-agent synchronization problem has remained largely unexplored.

 In this paper, we consider the attitude synchronization problem for a group of heterogeneous rigid body systems evolving on $SO(3)$ under an undirected, acyclic and connected graph topology. We begin by presenting a continuous distributed attitude synchronization scheme with almost global asymptotic stability guarantees. This scheme is based on the gradient of a smooth potential function defined on the rotation manifold. However, due to the topological obstruction to global asymptotic stability on $SO(3)$ via continuous time-invariant feedback, this potential function admits undesired critical points that prevent the continuous scheme from achieving global asymptotic stability. To overcome this limitation, we construct a new potential function by augmenting the original design with scalar variables. 
 The scalar variables are dynamically updated through a multi-agent switching mechanism, resulting in a hybrid feedback control scheme making the desired configuration the only attractor and leading to global attitude synchronization.
The proposed hybrid feedback strategy relies on the angular velocity measurements as well as the relative attitude information. Eliminating the need for velocity measurements in a network with a large number of agents can significantly reduce the cost associated with sensors and the communication flow between agents. Additionally, it ensures a certain level of immunity against angular velocity sensor failures. Moreover, it enhances robustness against potential angular velocity sensor failures. Motivated by the auxiliary systems approach introduced in \cite{Tayebi_TAC2008,miaomiao_TAC2022}, we develop a velocity-free distributed hybrid feedback control law for attitude synchronization that depends solely on attitude measurements, while guaranteeing global asymptotic stability. 
 This velocity-free law uses the outputs of some auxiliary dynamical systems to generate the necessary damping to compensate for the lack of angular velocity information. To the best of the author's knowledge, these are the first results in the literature dealing with global attitude synchronization on $SO(3)$ with and without angular velocity measurements. A preliminary version of this work was presented in \cite{Mouaad_ACC24}. However, that version did not include the continuous distributed attitude synchronization scheme or the hybrid distributed attitude synchronization scheme presented in this paper. Furthermore, compared to the preliminary conference version \cite{Mouaad_ACC24}, this paper provides a more thorough stability analysis, supplemented by further discussions, critical remarks, and more illustrative simulations.

The remainder of this paper is structured as follows: Section \ref{s2} provides some preliminaries and the notations used in this paper. Section \ref{s3} defines the problem of attitude synchronization on $SO(3)$. Section \ref{s34} introduces the continuous distributed attitude synchronization scheme. Section \ref{s4} introduces a generic potential function on $SO(3)^M \times \mathbb{R}^M$ used to derive our results. In Section \ref{s5}, we propose a generic distributed hybrid feedback control scheme for global attitude synchronization, which relies on angular velocity and relative attitude measurements. This scheme is further extended in Section \ref{s6} to eliminate the requirement for velocity measurements. Finally, Sections \ref{s8} and \ref{s9} present simulation results and concluding remarks, respectively.

\section{Preliminaries}\label{s2}

\subsection{Notation}
\noindent The sets of real numbers and the $n$-dimensional Euclidean space  are denoted by $\mathbb{R}$ and $\mathbb{R}^n$, respectively. The set of unit vectors in $\mathbb{R}^n$ is defined as $\mathbb{S}^{n-1}:=\{x\in \mathbb{R}^n~|~x^\top  x =1\}$. Given two matrices $A$,$B$ $\in \mathbb{R}^{m\times n}$, their Euclidean inner product is defined as $\langle \langle A,B \rangle \rangle=\text{tr}(A^\top  B)$. The Euclidean norm of a vector $x \in \mathbb{R}^n$ is defined as $||x||=\sqrt{x^\top  x}$, and the Frobenius norm of a matrix $A \in \mathbb{R}^{n\times n}$ is given by $||A||_F=\sqrt{\text{tr}(A^\top  A)}$. The matrix $I_n \in \mathbb{R}^{n \times n}$ denotes the identity matrix, and $\textbf{1}_n=[1\hdots1]^\top  \in \mathbb{R}^n$. Consider a Riemannian manifold $\mathcal{Q}$ with $\mathcal{T}_x \mathcal{Q}$ being its tangent space at point $x \in \mathcal{Q}$. Let $f: \mathcal{Q} \rightarrow \mathbb{R}_{\geq 0}$ be a continuously differentiable real-valued function. The function $f$ is a potential function on $\mathcal{Q}$ with respect to set $\mathcal{B} \subseteq \mathcal{Q}$ if $f(x)=0$, $\forall x \in \mathcal{B}$, and $f(x) > 0$, $\forall x\notin \mathcal{B}$. The gradient of $f$ at $x \in \mathcal{Q}$, denoted by $\nabla_x f(x)$, is defined as the unique element of $\mathcal{T}_x \mathcal{Q}$ such that \cite{Mahony_book_OAMM}:
\begin{equation}
    \dot f(x)=\langle \nabla_x f(x), \eta\rangle_x, ~~~~~ \forall \eta \in \mathcal{T}_x \mathcal{Q},
\end{equation}
where $\langle~,~\rangle_x:\mathcal{T}_x \mathcal{Q} \times \mathcal{T}_x \mathcal{Q} \rightarrow \mathbb{R}$ is a Riemannian metric associated to $\mathcal{Q}$. The point $x \in \mathcal{Q}$ is said to be a critical point of $f$ if $\nabla_x f(x)=0$.

The attitude of a rigid body is represented by a rotation matrix $R$ which belongs to the special orthogonal group $SO(3):= \{ R\in \mathbb{R}^{3\times 3} | \hspace{0.1cm}\text{det}(R)=1, R^\top R=I_3\}$. The $SO(3)$ group has a compact manifold structure and its tangent space is given by $\mathcal{T}_RSO(3):=\{R \hspace{0.1cm}\Omega\in \mathbb{R}^{3\times 3} \hspace{0.2cm} | \hspace{0.2cm} \Omega \in \mathfrak{so}(3)\}$ where $\mathfrak{so}(3):=\{ \Omega \in \mathbb{R}^{3\times 3} | \Omega^\top =-\Omega\}$ is the Lie algebra of the matrix Lie group $SO(3)$. The map $[.]^{\times}: \mathbb{R}^3 \rightarrow \mathfrak{so}(3)$ is defined such that $[x]^\times y=x \times y$, for any $x,y \in \mathbb{R}^3$, where $\times$ denotes the vector cross product on $\mathbb{R}^3$. The inverse map of $[.]^{\times}$ is $\text{vex}: \mathfrak{so}(3) \rightarrow \mathbb{R}^3$ such that $\text{vex}([\omega]^\times)=\omega$, and $[\text{vex}(\Omega)]^\times=\Omega$ for all $\omega \in \mathbb{R}^3$ and $\Omega \in \mathfrak{so}(3)$. Also, let $\mathbb{P}_a : \mathbb{R}^{3\times 3} \rightarrow \mathfrak{so}(3)$ be the projection map on the Lie algebra $\mathfrak{so}(3)$ such that $\mathbb{P}_a(A):=(A-A^\top )/2$. Given a 3-by-3 matrix $C:=[c_{ij}]_{i,j=1,2,3}$, one has  $\psi(C) := \text{vex} \circ \mathbb{P}_a (C)=\text{vex}(\mathbb{P}_a(C))=\frac{1}{2}[c_{32}-c_{23},c_{13}-c_{31},c_{21}-c_{12}]^\top $. For any $R\in SO(3)$, the normalized Euclidean distance on $SO(3)$, with respect to the identity $I_3$, is defined as $|R|_I^2:=\frac{1}{4}\text{tr}(I_3-R)$ $\in[0,1]$. The angle-axis parameterization of $SO(3)$, is given by $\mathcal{R}(\theta, u):=I_3+\sin\hspace{0.05cm}\theta \hspace{0.2cm}[u]^\times + (1-\cos\hspace{0.05cm}\theta)([u]^\times)^2$, where $u\in \mathbb{S}^2$ and  $\theta \in \mathbb{R}$ are the rotational axis and angle, respectively. The Kronecker product of two matrices $A$ and $B$ is denoted by $A \otimes B$.

\subsection{Graph Theory}
Consider a network of $N$ agents. The interaction topology between the agents is described by an undirected graph $\mathcal G = (\mathcal V,\mathcal E)$, where $\mathcal V=\{1,...,N\}$ and $\mathcal E \subseteq \mathcal V \times \mathcal V $ represent the vertex (or agent) set and the edge set of graph $\mathcal{G}$, respectively. In undirected graphs, the edge $(i,j) \in \mathcal E$ indicates that agents $i$ and $j$ interact with each other without any restriction on the direction, which means that agent $i$ can obtain information (via communication, measurements, or both) from agent $j$ and vice versa. The set of neighbors of agent $i$ is defined as $\mathcal N_i = \{j \in \mathcal V : (i,j) \in \mathcal E \}$. The undirected path is a sequence of edges in an undirected graph. An undirected graph is called connected if there is an undirected path between every pair of distinct agents of the graph. An undirected graph has a cycle if there exists an undirected path that starts and ends at the same agent \cite{Ren_book}. An acyclic undirected graph is an undirected graph without a cycle. An undirected tree is an undirected graph in which any two agents are connected by exactly one path (\ie, an undirected tree is an undirected, connected, and acyclic graph). An oriented graph is obtained from an undirected graph by assigning an arbitrary direction to each edge \cite{Mesbahi_book}.

Consider an oriented graph where each edge is indexed by a number. Let $M=|\mathcal{E}|$ and $\mathcal{M}=\{1,\hdots, M\}$ be the total number of edges and the set of edge indices, respectively. The incidence matrix, denoted by $H\in\mathbb{R}^{N\times M}$, is defined as follows \cite{BAI20083170}:
\begin{equation}\label{h_matrix}
   H:=[h_{ik}]_{N\times M} \hspace{0.4cm} \text{with} \hspace{0.2cm} h_{ik}=\begin{cases}
      +1 & k\in\mathcal{M}_i^+\\
      -1 & k\in\mathcal{M}_i^-\\
      0 & \text{otherwise}
    \end{cases},    
\end{equation}   
where $\mathcal{M}_i^+ \subset \mathcal{M}$ denotes the subset of edge indices in which agent $i$ is the head of the edges and $\mathcal{M}_i^- \subset \mathcal{M}$ denotes the subset of edge indices in which agent $i$ is the tail of the edges. For a connected graph, one verifies that $H^\top \textbf{1}_N=0$ and $\text{rank}(H)=N-1$. Moreover, the columns of $H$ are linearly independent if the graph is an undirected tree.

\subsection{Hybrid Systems Framework}
A hybrid system consists of continuous dynamics called flows and discrete dynamics called jumps. Given a manifold $\mathcal{Y}$ embedded in $\mathbb{R}^n$, according to \cite{Goebel_automatica_2006, Goebel_ieee_magazine,goebel2012hybrid}, the hybrid system dynamics, for every $y\in \mathcal{Y}$, is given by:
\begin{align}\label{H_sys}
    \mathcal{H}:
    \begin{cases}
       \dot{y} \in F(y) &\quad y\in \mathcal{F}\\
       y^+ \in G(y) & \quad y\in \mathcal{J}
    \end{cases}
\end{align}
where $F$ and $G$ represent the flow and jump maps, respectively, which govern the dynamics of the state $y$ through a continuous flow (if $y$  belongs to the flow set $\mathcal{F}$) and a discrete jump (if  $y$ belongs to the jump set $\mathcal{J}$). Note that $F:\mathcal{Y} \rightrightarrows \mathcal{T}\,\mathcal{Y}$ and $G:\mathcal{Y} \rightrightarrows \mathcal{Y}$ are set-valued mappings, where $\mathcal{T}\,\mathcal{Y}$ denotes the \textit{tangent bundle}\footnote{The disjoint union of all tangent spaces represents the \textit{tangent bundle} of $\mathcal{Y}$, \ie, $\mathcal{T}\,\mathcal{Y}=\bigcup_{y\in \mathcal{Y}}  \mathcal{T}_y \mathcal{Y}.$} of $\mathcal{Y}$. According to the nature of the hybrid system dynamics, which allows continuous flows and discrete jumps, the solutions of the hybrid system are parameterized by $t \in \mathbb{R}_{\geq 0}$ to indicate the amount of time spent in the flow set and $j \in \mathbb{N}$ to track the number of jumps that occur. The structure that represents this parameterization, which is known as a \textit{hybrid time domain}, is a subset of $\mathbb{R}_{\geq 0} \times \mathbb{N}$ and is denoted as \textup{dom}$\,y$. If the solution of the hybrid system $\mathcal{H}$ cannot be extended by either flowing or jumping, it is called \textit{maximal}, and \textit{complete} if its domain \textup{dom}$\,y$ is unbounded.

\section{Problem Statement}\label{s3}
Consider $N$-agent system governed by the following rigid-body rotational dynamics:
\begin{align}\label{R_dynamics_i}
    \begin{cases}
       \dot{R}_i &= R_i[\omega_i]^{\times}\\
       J_i \dot{\omega}_i &= -[\omega_i]^\times J_i \omega_i + \tau_i
    \end{cases}
\end{align}
with $R_i \in SO(3)$ represents the orientation of the body-attached frame of agent $i$ with respect to the inertial frame, $\omega_i\in \mathbb{R}^3$ is the body-frame angular velocity of agent $i$, $\tau_i \in \mathbb{R}^3$ is the control torque that will be designated later, and $J_i \in \mathbb{R}^{3 \times 3}$ is the inertial matrix associated with agent $i$. We emphasize that the agents are allowed to have different inertia matrices $J_i$, which are constant and symmetric positive definite.\\
Let the graph $\mathcal G$ describe the interaction between agents, which implies that if two agents are neighbors, their relative orientation is available to each of them, either by measurement if the agents are equipped with a relative attitude sensor, or by construction if they share their absolute orientations through communication. The relative orientation between agent $i$ and agent $j$ is defined as follows
\begin{equation}\label{measurement_model_R}
    R_{ij} := R_i^\top  R_j,
\end{equation}
where $(i,j)\in\mathcal E$. In this work, the interaction graph $\mathcal G$ is assumed to be an undirected tree. Now, let us formally introduce the problem that will be addressed in this work. 

\textbf{\textit{Problem}}: Consider a network of $N$ agents rotating according to the rigid-body rotational dynamics given in \eqref{R_dynamics_i}. Assume that the measurement \eqref{measurement_model_R} is available and the interaction graph $\mathcal G$ is an undirected tree. For each $i\in \mathcal{V}$, design a distributed feedback control torque $\tau_i$ such that, for any initial conditions, the orientations of all agents are synchronized to a common constant orientation.

Based on the above problem statement, the objective of this work is to design $\tau_i$ such that $R_j^T R_i = I_3$ and $\omega_i = 0$ for all $i \in \mathcal{V}$ and $j \in \mathcal{N}_i$, starting from any initial conditions.

\section{Distributed Attitude Synchronization on $SO(3)$}\label{s34}
 In this section, we present a continuous distributed attitude synchronization scheme. This discussion aims to provide context and motivation for the main results of the present paper, complemented by a rigorous stability analysis that demonstrates the scheme's ability to ensure almost global asymptotic stability. For every $i \in \mathcal{V}$, consider the following distributed feedback control law
\begin{equation}\label{continuous_tau} 
    \tau_i = - k_R \sum_{j \in \mathcal{N}_i} \psi(A R_j^\top R_i)-k_\omega \omega_i -\bar k_\omega \sum_{j\in\mathcal{N}_i}(\omega_i-\omega_j),
\end{equation}
where $k_R, k_\omega >0$, $\bar k_\omega \geq 0$ and $A \in \mathbb{R}^{3 \times 3}$ is a symmetric and positive definite matrix with three distinct eigenvalues. Note that the continuous distributed feedback control scheme \eqref{continuous_tau} has been previously proposed in \cite{SARLETTE2009572} with $A=I_3$ and $\bar k_\omega=0$, where only local asymptotic stability was claimed.
The first term in the control law \eqref{continuous_tau} is derived from the gradient of a potential function in terms of the orientation mismatches between neighboring agents (see the proof of Theorem \ref{theorem_continuous} for details). The second term acts as a local damping mechanism that ensures exponential convergence of each agent’s angular velocity to zero. The third term facilitates coordination by promoting consensus among the agents' angular velocities during the transient phase, resulting in more collective rotational behavior. In other words, the continuous feedback control law \eqref{continuous_tau} is designed to ensure that, for every \(i \in \mathcal{V}\) and \(j \in \mathcal{N}_i\), the relative orientation \(R_j^\top R_i\) converges to the identity matrix, and the angular velocity \(\omega_i\) converges to zero.

It is important to note that each edge in the graph has two possible relative orientations (\ie, for every \((i, j) \in \mathcal{E}\), both \(R_j^\top R_i\) and \(R_i^\top R_j\) are defined). However, the convergence of one inherently guarantees the convergence of the other. To simplify the stability analysis and eliminate redundancy, we consider only one relative orientation for each edge. This is achieved by assigning a virtual (arbitrary) orientation to the graph \(\mathcal{G}\) and indexing each oriented edge with an integer number. Accordingly, for any two agents \(i\) and \(j\) connected by an oriented edge \(k\), the relative attitude is defined as \(\bar{R}_k := R_j^\top R_i\), where \(\{k\} = \mathcal{M}_i^+ \cap \mathcal{M}_j^- \in \mathcal{M}\). It follows from \eqref{R_dynamics_i} that
\begin{align}
       \dot{\bar R}_k &= \bar R_k[\bar \omega_k]^{\times} \label{R_bar_dynamics_k}\\
       J_i \dot{\omega}_i &= -[\omega_i]^\times J_i \omega_i + \tau_i, \label{w_dynamics_k}
\end{align}
where $i \in \mathcal{V}$, \(\{k\} = \mathcal{M}_i^+ \cap \mathcal{M}_j^-\) and $\bar \omega_k :=  \omega_i - \bar R_k^\top  \omega_j$ for every $(i, j) \in \mathcal{E}$. Note that, for every $i \in \mathcal{V}$ and $ j \in \mathcal{N}_i$, the intersection between the sets $\mathcal{M}_i^+$ and $\mathcal{M}_j^-$ is either a subset of $\mathcal{M}$ with a single element (if agent $i$ and agent $j$ are the head and tail, respectively, of the directed edge connecting them) or an empty set otherwise. Let $\bar \omega =[\bar \omega_1^\top , \bar \omega_2^\top , \hdots, \bar \omega_M^\top ]^\top  \in \mathbb{R}^{3M}$ and $\omega =[\omega_1^\top , \omega_2^\top , \hdots, \omega_N^\top ]^\top  \in \mathbb{R}^{3N}$. One can verify that \cite{BAI20083170}
\begin{equation}\label{s_bar}
    \bar \omega = \Bar{H}(t)^\top  \omega,
\end{equation}
where the time-varying matrix $\bar H$ is defined as follows:
\begin{equation}\label{H_bar}
   \Bar{H}(t):=[\Bar{h}_{ik}]_{N\times M} \hspace{0.3cm} \text{with} \hspace{0.3cm} \Bar{h}_{ik}=\begin{cases}
      I_3 & k\in\mathcal{M}_i^+\\
      -\bar R_k & k\in\mathcal{M}_i^-\\
      0 & \text{otherwise}
    \end{cases}.   
\end{equation} 
The matrix $\bar H(t)$ is obviously influenced by the interaction graph $\mathcal{G}$ and its orientation. Note that the arbitrary orientation assigned to the graph $\mathcal{G}$ is only a dummy orientation and does not change the nature of the interaction graph $\mathcal{G}$ from being an undirected graph.

By assigning a relative attitude to each edge based on the virtual orientation introduced to the graph, the objective of this work can be reformulated as designing a feedback law $\tau_i$ for each $i \in \mathcal{V}$, such that the equilibrium $(\bar{R}_1 = I_3, \dots, \bar{R}_M = I_3, \omega_1 = 0, \dots, \omega_N = 0)$ is globally asymptotically stable.

Before introducing the design of the two proposed distributed hybrid feedback schemes, we first present the following theorem. This theorem establishes the stability properties of the dynamics described in \eqref{R_bar_dynamics_k}-\eqref{w_dynamics_k} under the continuous control torque \eqref{continuous_tau}. It also highlights the limitations of using a continuous feedback control scheme on the rotation manifold and characterizes the equilibrium set of these dynamics.

Define $\mathcal{A}_z:=\{z\in \mathcal{S}_z: \forall k\in \mathcal{M}, \forall i \in \mathcal{V}, \hspace{0.1cm} \bar R_k=I_3, \omega_i=0\}$, where $z:= \left(\bar R_1, \hdots, \bar R_M, \omega_1, \hdots, \omega_N\right)\in \mathcal{S}_z$ with $\mathcal{S}_z:=SO(3)^M\times\mathbb{R}^{3N}$. Our first main result can be stated as follows:

\begin{thm}\label{theorem_continuous}
    Let a network of $N$ agents rotate according to the rigid-body rotational dynamics given in \eqref{R_dynamics_i}. Assume that the measurement \eqref{measurement_model_R} is available and the interaction graph $\mathcal G$ is an undirected tree. Consider the dynamics \eqref{R_bar_dynamics_k}-\eqref{w_dynamics_k} with control torque \eqref{continuous_tau}. Then, the following statements hold:
    \begin{enumerate}[i)]
    \item All solutions of \eqref{R_bar_dynamics_k}-\eqref{w_dynamics_k} with \eqref{continuous_tau} converge to the set of equilibria $\Upsilon_z : =\mathcal{A}_z \cup \{z \in \mathcal{S}_z : \bar{R}_m = I_3, \, \bar R_n=\mathcal{R}(\pi, u_{\beta_n}), \, \omega_i=0, \, \forall m \in \mathcal{M}^I, \, \forall n \in \mathcal{M}^\pi, \, \forall i \in \mathcal{V} \}$, where $\mathcal{M}^I \cup \mathcal{M}^\pi=\mathcal{M}$, $|\mathcal{M}^\pi|>0$, $|\mathcal{M}^I|\geq0$, $\beta_n \in \{1, 2, 3\}$, and $u_{\beta_n} \in \mathcal{E}(A)$ with $\mathcal{E}(A)\subset \mathbb{S}^2$ denotes the set of unit eigenvectors of matrix $A$.\label{set_of_equilibrium}
    \item The set of all undesired equilibrium points $\Upsilon_z \setminus \mathcal{A}_z$ is unstable.\label{unstability_of_equilibrium}
    \item The desired equilibrium set $\mathcal{A}_z$ is \textit{almost globally asymptotically stable}.\footnote{ The set $\mathcal{A}_z$ is said to be almost globally asymptotically stable if it is asymptotically stable, and attaractive from all initial conditions except a set of zero Lebesgue measure.} \label{stability_of_equilibrium}
    \end{enumerate}
\end{thm}

\begin{proof}
    See Appendix \ref{app_1}
\end{proof}

\begin{rmk}
    The matrix $A$ is required to be symmetric positive definite with three distinct eigenvalues to establish the result of Theorem \ref{theorem_continuous}. This condition is not strictly necessary for implementing the control law \eqref{continuous_tau}; in fact, previous works (\eg, \cite{SARLETTE2009572}) have used $A=I_3$ (which has repeated eigenvalues). However, the assumption of distinct eigenvalues is crucial in our stability analysis, as it ensures that the undesired critical points of the associated potential function are isolated and occur only at rotations of angle $\pi$ about the eigenvectors of $A$. This property allows us to rigorously establish almost global asymptotic stability and characterize the instability of undesired equilibria in Theorem \ref{theorem_continuous}. Note that almost global asymptotic stability is the strongest stability result that can be achieved with smooth vector fields on the rotation manifold, as discussed in \cite{Bhat_SCL2000}.
\end{rmk}

\section{Distributed Hybrid Feedback for Global Attitude Synchronization on $SO(3)$}\label{s5}
It follows from the discussion in the previous section that the trajectories of the dynamics \eqref{R_bar_dynamics_k}–\eqref{w_dynamics_k}, under the continuous control torque \eqref{continuous_tau}, may converge to a level set that includes the undesired equilibrium set $\Upsilon_z \setminus \mathcal{A}_z$. This behavior stems from the topological constraints of $SO(3)$, which preclude the existence of any continuous feedback control scheme that can guarantee global asymptotic stability of the desired equilibrium set $\mathcal{A}_z$ \cite{Koditschek,Bhat_SCL2000}. To overcome this limitation, we will subsequently propose two hybrid feedback control schemes, one with angular velocity measurement and one without, that guarantee global asymptotic stability. Since the two proposed hybrid distributed attitude synchronization schemes are designed in terms of a generic potential function, we first introduce this function in the following subsection along with some of its interesting properties, which will be central to the hybrid design.

\subsection{Generic Potential Function on $SO(3)^M\times \mathbb{R}^M$}\label{s4}
As shown in the proof of Theorem \ref{theorem_continuous}, the continuous distributed feedback scheme \eqref{continuous_tau} was derived using the gradient of a potential function composed of a smooth potential function defined on the rotation manifold. This potential function on $SO(3)$ is the weighted trace function, $\text{tr}(A(I-R))$, where $A = A^\top > 0$ with distinct eigenvalues and $R \in SO(3)$. This function is well-established in the literature and has been widely used in feedback control and estimation design on $\mathrm{SO}(3)$; see, for example, \cite{mahony_tac2008,Mayhew_ACC2011}. However, as discussed earlier, designing gradient-based control laws using smooth potential functions on $SO(3)$ often leads to non-global results. To address this issue, numerous authors have proposed hybrid gradient-based solutions that guarantee the existence of a unique global attractor \cite{Mayhew_ACC2011, Mayhew_ACC2011_2, Mayhew_TAC2013, 7462234}. The core idea behind these solutions is to employ a switching mechanism that alternates between a family of smooth potential functions, generating a non-smooth gradient with a single global attractor. However, the construction of this family of smooth potential functions depends on the compactness of the manifold, rendering these approaches inapplicable to non-compact manifolds. Recently, the authors of \cite{miaomiao_TAC2022} introduced a novel hybrid scheme that overcomes this limitation. Their approach relies on a single potential function parameterized by a scalar variable governed by hybrid dynamics. By appropriately switching the value of this variable, the potential function is adjusted such that the resulting non-smooth gradient possesses only one global attractor. Unlike the methods in \cite{Mayhew_ACC2011, Mayhew_ACC2011_2, Mayhew_TAC2013, 7462234}, the hybrid scheme in \cite{miaomiao_TAC2022} is simpler to implement and does not require any compactness assumptions on the manifold, making it applicable even to non-compact manifolds such as $SE(3)$. In the present work, inspired by the hybrid approach on $SO(3)$ in \cite{miaomiao_TAC2022}, we design a new potential function on $SO(3)^M \times \mathbb{R}^M$ suitable for application of hybrid techniques for multi-agent systems evolving on \( SO(3) \) leading to global asymptotic stability guarantees.

Let $\mathcal{A}:=\{x\in \mathcal{S}: \forall k\in \mathcal{M}, \hspace{0.1cm} R_k=I_3, \xi_k=0\}$, where $x:= \left(R_1, \hdots, R_M, \xi_1, \hdots, \xi_M\right)\in \mathcal{S}$ with $\mathcal{S}:=SO(3)^M\times\mathbb{R}^M$. Consider the following potential function, on $\mathcal{S}$, with respect to $\mathcal{A}$:
\begin{align}\label{potential_fct}
    \bar U(x) &:= \sum_{k=1}^{M}U( R_k,\xi_k),
\end{align} 
where $U: SO(3)\times \mathbb{R} \rightarrow \mathbb{R}_{\geq0}$ is a potential function with respect to $(I_3, 0)$. The following set represents the set of all critical points of $\bar U(x)$:
\begin{align}
    \bar \Upsilon := \{x \in \mathcal{S}: \forall k\in \mathcal{M}, \nabla_{R_k}\bar U(x)=0,~\nabla_{\xi_k}\bar U(x)=0\},\nonumber
\end{align}
where $\nabla_{R_k}\bar U(x)$ and $\nabla_{\xi_k}\bar U(x)$ are the gradients of $\bar U(x)$ with respect to $R_k$ and $\xi_k$, respectively.  The potential function $\bar U$ is chosen such that $\mathcal{A} \subset \bar \Upsilon$.

Next, inspired by \cite{miaomiao_TAC2022}, we will introduce an instrumental condition to our proposed hybrid feedback control design.
\begin{cnd}\label{hybrid_ass}
    Consider the potential function \eqref{potential_fct}. There exist a scalar $\delta_{\bar R}>0$ and a nonempty finite set $\Xi \subset \mathbb{R}$ such that for every $x= \left(R_1, \hdots, R_M, \xi_1, \hdots, \xi_M\right)\in \bar \Upsilon \setminus \mathcal{A}$
     \begin{equation}\label{cond1}
        U(R_k, \xi_k) - \underset{\bar{\xi}_k \in \Xi}{\text{min}}~U(R_k, \bar{\xi}_k) > \delta_{\bar R},
    \end{equation}
   for every $k \in \mathcal{M}$ such that $R_k \neq I_3$.
\end{cnd}
\begin{rmk}

\indent The set $\bar{\Upsilon} \setminus \mathcal{A}$ is the set of all undesired critical points of $\bar{U}(x)$. Inequality \eqref{cond1} indicates that whenever the state belongs to  $\bar \Upsilon \setminus \mathcal{A}$, there will always exist $\bar \xi_k \in \Xi$, for every $k \in \mathcal{M}$, where $R_k \neq I_3$, such that $U(R_k, \bar{\xi}_k)$ is lower than $U(R_k, \xi_k)$ by a constant gap $\delta_{\bar R}$. This fact will be used to design the switching mechanism associated with the hybrid parameters $\xi_k$, for every $k \in \mathcal{M}$, in our proposed hybrid schemes.
\end{rmk}
To satisfy Condition \ref{hybrid_ass}, one can, for instance, consider the potential function $\bar U(x)$ in \eqref{potential_fct}, where $U(R_k, \xi_k)$ is defined as follows:
\begin{equation}\label{expl_pf}
    U(R_k,\xi_k):=\text{tr}\Big(A\left(I_3-R_k\mathcal{R}(\xi_k,u)\right)\Big)+\frac{\gamma}{2}\xi_k^2,
\end{equation}
where $A \in \mathbb{R}^{3 \times 3}$ is a symmetric and positive definite matrix with three distinct eigenvalues, $u\in \mathbb{S}^2$ is a constant unit vector and $\gamma$ is a positive scalar. 
The potential function $U$ was introduced in \cite{miaomiao_TAC2022} for the global attitude tracking via hybrid feedback. Note that the potential function \eqref{expl_pf} is based on the weighted trace function, where $U(R_k,0)=\text{tr}\left(A(I_3-R_k)\right)$. The hybrid switching parameter $\xi_k$ appears in the first term as the angle of rotation about a fixed unit vector $u$ applied to $\bar{R}_k$, and in the second term as a quadratic term. As we will demonstrate in the sequel, these parameters $\xi_k$, for all $k \in \mathcal{M}$, are introduced to shape the potential function $\bar U$ in a way that ensures it admits a unique global critical point corresponding to the desired synchronization configuration. Define the set of parameters $\mathcal{P}:=\{\Xi, A, u, \gamma, \delta_{\bar R} \}$. The next proposition gives the possible choices of parameters in the set $\mathcal{P}$ in which Condition \ref{hybrid_ass} is satisfied.
\begin{pro}\label{pro_set}
	Consider the potential function \eqref{potential_fct} with $U( R_k,\xi_k)$ given in \eqref{expl_pf}. Then, Condition 1 holds under the following set of parameters $\mathcal{P}$:
\begin{equation} {\mathcal{P}}: \begin{cases} {\Xi = \left\{ {\left| {{\phi _i}} \right| \in (0,\pi ],i = 1, \cdots ,l} \right\}} \\ {A:0 < \lambda_1 \leq \lambda_2 < \lambda _3} \\ {u = {\alpha _1}q_1 + {\alpha _2}q_2 + {\alpha _3}q_2} \\ {\gamma < \frac{{4{\Delta ^{\ast}}}}{{{\pi ^2}}}} \\ {0<\delta_{\bar R} < \left( {\frac{{4{\Delta ^{\ast}}}}{{{\pi ^2}}} - \gamma } \right)\frac{{\phi_L^2}}{2},{\phi_L}: = \mathop {\max }\limits_{\phi \in \Xi } \left| \phi \right|} \end{cases}\end{equation}
where $\alpha_1^2+\alpha_2^2+\alpha_3^2=1$ and $\Delta^*>0$ are given as follows:
\begin{itemize}
    \item If $\lambda_1=\lambda_2$, $\alpha_3^2=1-\frac{\lambda_2}{\lambda_3}$ and $\Delta^*=\lambda_1(1-\frac{\lambda_2}{\lambda_3})$.
    \item If $\lambda_2\geq \frac{\lambda_1 \lambda_3}{\lambda_3-\lambda_1}$, $\alpha_i^2=\frac{\lambda^A_i}{\lambda_2+\lambda_3}$,$i\in \{2,3\}$ and $\Delta^*=\lambda_1$.
    \item If $\lambda_1<\lambda_2<\frac{\lambda_1\lambda_3}{\lambda_3-\lambda_1}$, $\alpha_i^2=1-\frac{4\prod_{l\neq i}\lambda_l}{\sum_{l=1}^{3}\sum_{k\neq l}^{3}\lambda_l\lambda_k}, \forall i\in \{1, 2, 3\}$, and $\Delta^*=\frac{4\prod_{l}\lambda_l}{\sum_{l=1}^{3}\sum_{k\neq l}^{3}\lambda_l\lambda_k}$.
\end{itemize}
where $(\lambda_i, q_i)$ denotes the $i$-th pair of eigenvalue-eigenvector of matrix $A$.
\end{pro}
\begin{proof}
Following the same arguments given in the proof of \cite[Proposition 2]{miaomiao_TAC2022}, one can prove Proposition \ref{pro_set}.
\end{proof}
\begin{rmk}
    As shown in \cite[Proposition 2]{Berkane_TAC2016}, the design of $u$ given in Proposition \ref{pro_set} is based on the following optimization $\max_{u \in \mathbb{S}^2} \left( \min_{v \in \mathcal{E}(A)} \Delta(v,u) \right)$, where $\Delta(v,u) = u^\top \left( \text{tr}(A) I_3 - A - 2 v^\top A v (I_3 - v v^\top) \right) u$.
\end{rmk}

\subsection{Distributed Hybrid Feedback for Global Attitude Synchronization on $SO(3)$ with Velocity Measurements}\label{s5}
In this section, we present the first proposed distributed hybrid feedback control scheme, which ensures global asymptotic synchronization of the agents' attitudes using relative attitude and angular velocity measurements. This scheme is developed based on the gradient of the generic potential function introduced in the previous section. 
For every $i \in \mathcal{V}$, we propose the following distributed hybrid feedback control scheme

{\small
\begin{align}
             &\underbrace{
                  \begin{aligned}\label{R_obs_f}
                    \tau_i =&k_R\bigg(\sum_{l\in \mathcal{M}_i^-} \bar R_l \psi \left(\bar R_l^\top  \nabla_{\bar R_l}\bar U\right)-\hspace{-0.13cm}\sum_{n\in \mathcal{M}_i^+} \psi \left(\bar R_n^\top  \nabla_{\bar R_n}\bar U\right)\bigg)\\
                    &-k_\omega \omega_i -\bar k_\omega \sum_{j\in\mathcal{N}_i}(\omega_i-\omega_j)\\
                    \dot{\xi}_k =&-k_\xi \nabla_{\xi_k}\bar U
                 \end{aligned}
                }_{x \in \mathcal{F}_i}
                \\
             &\underbrace{
                \begin{aligned}\label{R_obs_j}
                    \xi_k^+&\in 
                    \begin{cases}
                    \xi_k \hspace{1.3cm}\text{if} \hspace{0.5cm} U(\bar{R}_k, \xi_k)-U(\bar{R}_k, \xi^*_k) \leq \delta_{\bar R}\\
                    \xi^*_k \hspace{1.3cm}\text{if} \hspace{0.5cm} U(\bar{R}_k, \xi_k)-U(\bar{R}_k, \xi^*_k) \geq \delta_{\bar R}
                    \end{cases}
                \end{aligned}
                }_{x \in \mathcal{J}_i}
\end{align}}where $k_\xi, k_R, k_\omega >0$, $\bar k_\omega \geq 0$, $\xi^*_k := \text{arg} \underset{\bar{\xi}_k \in \Xi}{\text{min}} U(\bar{R}_k, \bar{\xi}_k)$ and $k \in \mathcal{M}_i^+$. The flow set $\mathcal{F}_i$ and the jump set $\mathcal{J}_i$, for agent $i$, are defined as follows:
{\small
\begin{align}    
    \mathcal{F}_i&:=\{x\in \mathcal{S} : \forall k\in \mathcal{M}_i^+,\hspace{0cm} U(\bar{R}_k,\xi_k)-\underset{\bar{\xi}_k\in \Xi}{\text{min}} U(\bar{R}_k,\bar{\xi}_k)\leq \delta_{\bar R}\},\nonumber\\
    \mathcal{J}_i&:=\{x \in \mathcal{S}: \exists k\in \mathcal{M}_i^+, 
    \hspace{0cm} U(\bar{R}_k,\xi_k)-\underset{\bar{\xi}_k\in \Xi}{\text{min}} U(\bar{R}_k,\bar{\xi}_k)\geq \delta_{\bar R}\}.\nonumber
\end{align}}It is clear from the definitions of the flow set \( \mathcal{F}_i \) and the jump set \( \mathcal{J}_i \), for each \( i \in \mathcal{V} \), that the constraints characterizing these sets are distributed, as they depend only on the edge states (\ie, $\bar R_k$ and $\xi_k$) where the agent $i$ is the head of the oriented edge (\ie, $k \in \mathcal{M}_i^+$). Similar to the continuous feedback scheme \eqref{continuous_tau}, the control torque, given in  \eqref{R_obs_f}, incorporates gradient-based terms (the first two) derived from the parametrized potential function \eqref{potential_fct}, which drive the system toward attitude synchronization. The remaining two terms facilitate coordinated convergence to a common constant orientation. However, unlike \eqref{continuous_tau}, the  control torque in \eqref{R_obs_f} employs hybrid variables with a switching mechanism \eqref{R_obs_j} that prevents the gradient-based terms from vanishing at undesired configurations. This ensures global attitude synchronization to a common orientation.

\begin{rmk}
The dynamics of the hybrid variable $\xi_k$, for each $ k \in \mathcal{M}$, are assumed to be executed by agent $i$, while agent $j$ receives information about $\xi_k$ from agent $i$ via communication, for every $(i,j) \in \mathcal{E}$ such that ${k} = \mathcal{M}_i^+ \cap \mathcal{M}_j^-$. The virtual (arbitrary) orientation assigned to the graph $\mathcal{G}$ provides a uniform way to implement the dynamics of the auxiliary states $\xi_k$ on the agents for each $k \in \mathcal{M}$. However, for distributed implementation, relying on arbitrary head–tail assignments may hinder scalability.
To ensure consistent and scalable deployment, we propose the following local rule: whenever a new agent joins the network, it is assigned the role of the \emph{head} of each newly formed edge and is therefore responsible for implementing the corresponding \(\xi_k\) dynamics. This assignment requires only local agreement between neighboring agents and can be applied consistently without global coordination.
This approach is particularly effective in tree-structured interaction graphs, where each new agent connects via exactly one edge. As a result, each agent is responsible for at most one \(\xi_k\) variable, and computational responsibilities are naturally distributed across the network. Moreover, this rule prevents the overloading of central agents in topologies such as stars and enables the controller to scale efficiently as the network grows. 
\end{rmk}

Define the new state $\bar x := \left(x, \omega_1, \hdots, \omega_N\right)\in \bar{\mathcal{S}}$, where $\Bar{\mathcal{S}}:= SO(3)^M \times \mathbb{R}^M\times \mathbb{R}^{3N}$. In view of \eqref{R_bar_dynamics_k}-\eqref{w_dynamics_k} and \eqref{R_obs_f}-\eqref{R_obs_j}, one can derive the following  multi-agent hybrid dynamics:
\begin{equation}\label{hybrid_sys}
\bar{\mathcal{H}}:\begin{cases} {\dot{\bar x} = \bar F(\bar x),}&{\bar x \in \bar{\mathcal{F}}:=\{\bar x \in \bar{\mathcal{S}}: x \in \mathcal{F}\}} \\ {{\bar x^ + } \in \bar G(\bar x),}&{\bar x \in \bar{\mathcal{J}}:=\{\bar x \in \bar{\mathcal{S}}: x \in \mathcal{J}\}} \end{cases}
\end{equation}
where
\begin{eqnarray}\label{network_f_j_set}
\mathcal{F}:=\bigcap_{i=1}^{N}\mathcal{F}_{i},\qquad \mathcal{J}:=\bigcup_{i=1}^{N}\mathcal{J}_{i},
\end{eqnarray}
and
{\small
\begin{equation}
    \bar F(\bar x) := \hspace{-0.1 cm}\left[\hspace{-0.2 cm} {\begin{array}{c} {\bar R_1[\bar \omega_1]^{\times}} \\ \vdots\\ {\bar R_M[\bar \omega_M]^{\times}}\\{-k_\xi \nabla_{\xi_1}\bar U} \\ \vdots\\ {-k_\xi \nabla_{\xi_M}\bar U}\\{J_1^{-1} \left(-[\omega_1]^\times J_1 \omega_1+\tau_1\right)} \\ \vdots\\ {J_N^{-1} \left(-[\omega_N]^\times J_N \omega_N+\tau_N\right)} \end{array}}\hspace{-0.2 cm}\right],~\bar G(\bar x) := \hspace{-0.1 cm}\left[\hspace{-0.2 cm} {\begin{array}{c} {\bar R_1} \\ \vdots\\ {\bar R_M}\\{\{\xi_1, \xi_1^*\}} \\ \vdots\\ {\{\xi_M, \xi_M^*\}}\\{\omega_1} \\ \vdots\\ {\omega_N} \end{array}}\hspace{-0.2 cm}\right] \nonumber,
\end{equation}}with $\tau_i$ is given in \eqref{R_obs_f} for each $i \in \mathcal{V}$. From equations \eqref{hybrid_sys}-\eqref{network_f_j_set}, one can deduce that $\bar{\mathcal{F}}\cup \bar{\mathcal{J}} = \bar{\mathcal{S}}$ and the hybrid closed-loop system \eqref{hybrid_sys} is autonomous. The next lemma shows that the hybrid closed-loop system \eqref{hybrid_sys} is well-posed\footnote{See \cite[Definition 6.2]{goebel2012hybrid} for the definition of well-posedness.} by verifying the hybrid basic conditions given in \cite[Assumption 6.5]{goebel2012hybrid}.

\begin{lem}\label{lem_hbc}
    The hybrid closed-loop system (\ref{hybrid_sys}) satisfies the following hybrid basic conditions:
    \begin{enumerate}[i)]
        \item $\bar{\mathcal{F}}$ and $\bar{\mathcal{J}}$ are closed subsets;\label{hbc_1}
        \item $\bar F$ is outer semicontinuous and locally bounded relative to $\bar{\mathcal{F}}$,
           $\bar{\mathcal{F}} \in \text{dom}~ \bar F$ , and $\bar F(\bar x)$ is convex for every $\bar x \in \bar{\mathcal{F}}$;\label{hbc_2}
        \item $\bar G$ is outer semicontinuous and locally bounded relative to $\bar{\mathcal{J}}$
           and $\bar{\mathcal{J}} \in \text{dom}~\bar G$.\label{hbc_3}  
    \end{enumerate}
\end{lem}
\begin{proof}
Following similar arguments given in the proof of \cite[Lemma 2]{boughellaba2023distributed}, one can prove Lemma \ref{lem_hbc}.
\end{proof}

\begin{rmk}
 Condition \ref{hybrid_ass} implies that the set of all undesired critical points belongs to the jump set $\mathcal{J}$, \textit{i.e.,} $\bar \Upsilon \setminus \mathcal{A} \subset \mathcal{J}$. The jump map $\bar G$ will reset the states to values resulting in a decrease of $\bar U(x)$. 
\end{rmk}

Now, we will present our second main result.

\begin{thm}\label{theorem1}
    Let $k_R, k_\xi, k_\omega>0$ and $\bar k_\omega \geq 0$. Suppose Condition \ref{hybrid_ass} is satisfied. Then, the number of jumps of the multi-agent hybrid closed-loop system (\ref{hybrid_sys}) is finite and the set $\bar{\mathcal{A}}:=\{\bar x \in \bar{\mathcal{S}}: x \in \mathcal{A}, ~\omega=0\}$ is globally asymptotically stable for the multi-agent hybrid closed-loop system (\ref{hybrid_sys}). \label{th1_1}
\end{thm}
\begin{proof}
    See Appendix \ref{app_2}
\end{proof}

\begin{rmk}\label{rmk_tv_ori}
    According to Theorem~\ref{theorem1}, the hybrid distributed feedback control law \eqref{R_obs_f}-\eqref{R_obs_j} drives the agents' orientations from any initial condition to a common constant orientation. In contrast, Theorem~\ref{theorem_continuous} demonstrates that the distributed feedback control law \eqref{continuous_tau} achieves the same objective but only from almost any initial condition. However, when both schemes \eqref{R_obs_f}-\eqref{R_obs_j} and \eqref{continuous_tau} are applied with $k_\omega = 0$ and $\bar{k}_\omega > 0$, the agents' orientations converge to a common time-varying orientation, as illustrated in the Simulation section (Figures~\ref{R_contin_k_w_0} and~\ref{R_hybrid_k_w_0}).
\end{rmk}
\begin{rmk}
    The term $-k_\omega \omega_i$ in the proposed hybrid control scheme \eqref{R_obs_f} ensures the convergence of $\omega_i$ to zero, which is essential for achieving the result of Theorem \ref{theorem1}. On the other hand, the last term $-\bar k_\omega \sum_{j\in\mathcal{N}_i}(\omega_i-\omega_j)$ is not necessary to prove the result established in Theorem \ref{theorem1}, as the theorem remains valid even when $\bar k_\omega = 0$. However, this term provides inter-agent damping that contributes to transient performance improvement.
\end{rmk}

The distributed hybrid feedback control law \eqref{R_obs_f}-\eqref{R_obs_j} was developed based on a generic potential function defined on $SO(3)^M \times \mathbb{R}^M$. In the following, we derive the explicit forms of the feedback control law \eqref{R_obs_f}-\eqref{R_obs_j} using a specific potential function. Considering the potential function defined in \eqref{expl_pf}, the explicit form of the distributed hybrid feedback law \eqref{R_obs_f}-\eqref{R_obs_j} is obtained as follows:

{\small
\begin{align}
             &\underbrace{
                  \begin{aligned}\label{explicit_11}
                    \tau_i =&-k_R\bigg(\sum_{j \in \mathcal{O}_i} \psi\left(A \mathcal{R}_a(\xi_n, u)^\top  R_j^\top  R_i\right)+ \sum_{j \in \mathcal{I}_i} \mathcal{R}_a(\xi_p, u)\\
                    &\psi\left(A R_j^\top  R_i \mathcal{R}_a(\xi_p, u)\right)\bigg)-k_\omega \omega_i -\bar k_\omega \sum_{j\in\mathcal{N}_i}(\omega_i-\omega_j)\\
                    \dot{\xi}_k =&-k_\xi \left(\gamma \xi_k+2 u^\top  \psi\left(A \bar R_k \mathcal{R}_a(\xi_k, u)\right)\right)
                 \end{aligned}
                }_{x \in \mathcal{F}_i}
                \\
             &\underbrace{
                \begin{aligned}\label{explicit_12}
                    \xi_k^+&\in 
                    \begin{cases}
                    \xi_k \hspace{1.3cm}\text{if} \hspace{0.5cm} U(\bar{R}_k, \xi_k)-U(\bar{R}_k, \xi^*_k) \leq \delta_{\bar R}\\
                    \xi^*_k \hspace{1.3cm}\text{if} \hspace{0.5cm} U(\bar{R}_k, \xi_k)-U(\bar{R}_k, \xi^*_k) \geq \delta_{\bar R}
                    \end{cases}
                \end{aligned}
                }_{x \in \mathcal{J}_i}
\end{align}}where $i\in\mathcal{V}$, $ k \in \mathcal{M}_i^+$, $\{p\}=\mathcal{M}_i^+\cap \mathcal{M}_j^- \in \mathcal{M}$, $\{n\}=\mathcal{M}_i^-\cap \mathcal{M}_j^+ \in \mathcal{M}$, $\mathcal{I}_i:=\{j\in \mathcal{N}_i: j \hspace{0.1cm} \text{is the tail of the edge} \hspace{0.1cm} (i,j) \in \mathcal{E}\}$ and $\mathcal{O}_i:=\{j\in \mathcal{N}_i: j \hspace{0.1cm} \text{is the head of the edge} \hspace{0.1cm} (i,j) \in \mathcal{E}\}$. where $i \in \mathcal{V}$. From the fact that $\mathcal{N}_i = \mathcal{I}_i \cup \mathcal{O}_i$, it is clear that the hybrid feedback control law \eqref{explicit_11}-\eqref{explicit_12} is distributed in the sense that each agent relies solely on information from neighboring agents. Furthermore, the implementation of the proposed distributed hybrid feedback law, as described in equations \eqref{explicit_11}-\eqref{explicit_12}, depends only on relative attitude and angular velocity measurements. These measurements can be readily obtained using onboard sensors or through inter-agent communication within the network.

\subsection{Distributed Hybrid Feedback for Global Attitude Synchronization  on $SO(3)$ without Velocity Measurements}\label{s6}
The distributed hybrid feedback scheme presented in \eqref{R_obs_f}-\eqref{R_obs_j} requires each agent to have access to its angular velocity. However, this requirement can be resource-intensive, particularly in networks with a large number of agents. To address this challenge, we propose a velocity-free distributed hybrid synchronization scheme. This scheme introduces an auxiliary dynamic system for each agent, which generates the necessary damping to compensate for the lack of angular velocity measurements.

Before detailing the velocity-free distributed hybrid synchronization scheme, we first define the dynamics of the auxiliary states, $(Q_i, \zeta_i) \in SO(3) \times \mathbb{R}$, for each agent $i \in \mathcal{V}$ as follows:

{\small 
\begin{align}
             &\underbrace{
                  \begin{aligned}\label{q_11}
                    \dot{Q}_i =&k_Q Q_i\left[\Tilde{Q}_i \psi\left(\Tilde{Q}^\top_i \nabla_{\Tilde{Q}_i}U(\Tilde{Q}_i, \zeta_i)\right)\right]^\times\\
                    \dot{\zeta}_i =&-k_\zeta \nabla_{\zeta_i}U(\Tilde{Q}_i, \zeta_i)
                 \end{aligned}
                }_{(Q_i, \zeta_i) \in \mathcal{F}_i^{\tilde Q}}\\
             &\underbrace{
                \begin{aligned}\label{q_21}
                    Q_i^+&=Q_i\\
                    \zeta_i^+&\in 
                    \begin{cases}
                    \zeta_i \hspace{1.3cm}\text{if} \hspace{0.5cm} U(\Tilde{Q}_i, \zeta_i)-U(\Tilde{Q}_i, \zeta^*_i) \leq \delta_{\tilde Q}\\
                    \zeta^*_i \hspace{1.3cm}\text{if} \hspace{0.5cm} U(\Tilde{Q}_i, \zeta_i)-U(\Tilde{Q}_i, \zeta^*_i) \geq \delta_{\tilde Q}
                    \end{cases}
                \end{aligned}
                }_{(Q_i, \zeta_i) \in \mathcal{J}_i^{\tilde Q}}
\end{align}
}where $k_Q, k_\zeta > 0$, $Q_i(0) \in SO(3)$, $\zeta_i(0) \in \mathbb{R}$, $\Tilde{Q}_i:= Q^\top_i R_i$ and $\zeta^*_i := \text{arg} \underset{\bar{\zeta}_i \in \Pi}{\text{min}} U(\tilde Q_i, \bar{\zeta}_i)$ . After introducing the following condition, adopted from \cite{miaomiao_TAC2022}, we define the flow set $\mathcal{F}_i^{\tilde Q}$ and the jump set $\mathcal{J}_i^{\tilde Q}$ shown in \eqref{q_11}-\eqref{q_21}.
\begin{cnd}\label{hybrid_ass_2}
    Let $U$ be a potential function on $SO(3)\times\mathbb{R}$, with respect to $(I_3, 0)$. Let $(I_3, 0)\in \Upsilon$, where $\Upsilon:=\{(\tilde Q_i, \zeta_i) \in SO(3) \times \mathbb{R}:~\nabla_{\tilde Q_i}U(\tilde Q_i, \zeta_i)=0, \nabla_{\zeta_i}U(\tilde Q_i, \zeta_i)=0\}$ is the set of all critical points of $U(\tilde Q_i, \zeta_i)$. There exist a scalar $\delta_{\tilde Q}>0$ and a nonempty finite set $\Pi$ such that, for every $(\tilde Q_i, \zeta_i)\in \Upsilon \setminus \{(I_3, 0)\}$, one has
    \begin{equation}
    U(\tilde Q_i, \zeta_i)-\underset{\bar \zeta_i \in \Pi}{\text{min}} U(\tilde Q_i, \bar \zeta_i) > \delta_{\tilde Q}.
    \end{equation}
\end{cnd}
\begin{rmk}
The motivation behind Condition \ref{hybrid_ass_2} is similar to that of Condition \ref{hybrid_ass}. Condition \ref{hybrid_ass_2} implies that all undesired critical points in $\Upsilon \setminus \{ (I_3, 0) \}$ are inside the jump set $\mathcal{J}_i^{\tilde Q}$, and as such, the jump map in \eqref{q_21} will take care of steering the state away from the undesired critical points $\Upsilon \setminus \{ (I_3, 0) \}$.   
\end{rmk}

\begin{rmk}
    Consider the potential function $U$ defined in \eqref{expl_pf}. Proposition 2 in \cite{miaomiao_TAC2022} gives the possible choices of parameters $\{\Pi, A, u, \gamma, \delta_{\tilde Q}\}$ for which Condition \ref{hybrid_ass_2} is satisfied.
\end{rmk}

Based on Condition \ref{hybrid_ass_2}, for each $i \in \mathcal{V}$, one defines the flow set $\mathcal{F}_i^{\tilde Q}$ and the jump set $\mathcal{J}_i^{\tilde Q}$ as follows:
{\small
\begin{align}    
    \mathcal{F}_i^{\tilde Q}&\hspace{-0.1cm}:=\hspace{-0.1cm}\{(\tilde Q_i, \zeta_i)\in SO(3)\times \mathbb{R} \hspace{-0.1cm}:U(\tilde Q_i, \zeta_i)-\underset{\bar \zeta_i \in \Xi}{\text{min}} U(\tilde Q_i, \bar \zeta_i)\leq \delta_{\tilde Q}\},\nonumber\\
    \mathcal{J}_i^{\tilde Q}&\hspace{-0.1cm}:=\hspace{-0.1cm}\{(\tilde Q_i, \zeta_i)\in SO(3)\times \mathbb{R} \hspace{-0.1cm}:U(\tilde Q_i, \zeta_i)-\underset{\bar \zeta_i \in \Xi}{\text{min}} U(\tilde Q_i, \bar \zeta_i)\geq \delta_{\tilde Q}\}.\nonumber
\end{align}}Since $\Tilde{Q}_i:= Q^\top_i R_i$, it follows from \eqref{q_11}-\eqref{q_21} and \eqref{R_dynamics_i} that
{\small
\begin{align}
             &\underbrace{
                  \begin{aligned}\label{cl_q_1}
                    \dot{\tilde{Q}}_i =& \tilde{Q}_i \left[\omega_i-k_Q \psi\left(\Tilde{Q}^\top_i \nabla_{\Tilde{Q}_i}U(\Tilde{Q}_i, \zeta_i)\right)\right]^\times\\
                    \dot{\zeta}_i =&-k_\zeta \nabla_{\zeta_i}U(\Tilde{Q}_i, \zeta_i)
                 \end{aligned}
                }_{(Q_i, \zeta_i) \in \mathcal{F}_i^{\tilde Q}}\\
             &\underbrace{
                \begin{aligned}\label{cl_q_2}
                    \tilde Q_i^+ &= \tilde Q_i\\
                    \zeta_i^+&\in 
                    \begin{cases}
                    \zeta_i \hspace{1.3cm}\text{if} \hspace{0.5cm} U(\Tilde{Q}_i, \zeta_i)-U(\Tilde{Q}_i, \zeta^*_i) \leq \delta_{\tilde Q}\\
                    \zeta^*_i \hspace{1.3cm}\text{if} \hspace{0.5cm} U(\Tilde{Q}_i, \zeta_i)-U(\Tilde{Q}_i, \zeta^*_i) \geq \delta_{\tilde Q}
                    \end{cases}
                \end{aligned}
                }_{(Q_i, \zeta_i) \in \mathcal{J}_i^{\tilde Q}}
\end{align}}The primary objective of designing the auxiliary system \eqref{q_11}-\eqref{q_21} is to achieve an indirect asymptotic estimation of the angular velocity measurement for each agent. This compensates for the lack of angular velocities required in the control scheme \eqref{R_obs_f}-\eqref{R_obs_j} and thus ensures closed-loop stability without the need for angular velocity measurements. For further illustration, consider the hybrid closed-loop system \eqref{cl_q_1}-\eqref{cl_q_2}. It is clear that the convergence of $\tilde{Q}_i$ to $I_3$ inherently drives the term $\psi\left(\tilde{Q}^\top_i \nabla_{\tilde{Q}_i} U(\tilde{Q}_i, \zeta_i)\right)$ to asymptotically match $\omega_i$ for all $i \in \mathcal{V}$. As a result, this framework acts as an asymptotic observer for $\omega_i$, allowing $\omega_i$ to be replaced by $\psi\left(\tilde{Q}^\top_i \nabla_{\tilde{Q}_i} U(\tilde{Q}_i, \zeta_i)\right)$ to provide the necessary damping in the feedback control input $\tau_i$.

For every $i \in \mathcal{V}$, considering the auxiliary system \eqref{q_11}-\eqref{q_21}, we propose the following distributed hybrid velocity-free feedback control law:
{\small
\begin{align}
             &\underbrace{
                  \begin{aligned}\label{tau_w_v}
                    \tau_i =&k_R\bigg(\sum_{l\in \mathcal{M}_i^-} \bar R_l \psi \left(\bar R_l^\top  \nabla_{\bar R_l}\bar U\right)-\hspace{-0.13cm}\sum_{n\in \mathcal{M}_i^+} \psi \left(\bar R_n^\top  \nabla_{\bar R_n}\bar U\right)\bigg)\\
                    &-k_{\tilde Q}\psi\left(\tilde Q_i^\top  \nabla_{\tilde Q_i} U(\tilde Q_i, \zeta_i)\right)\\
                    \dot{\xi}_k =&-k_\xi \nabla_{\xi_k}\bar U
                 \end{aligned}
                }_{x \in \mathcal{F}_i}
                \\
             &\underbrace{
                \begin{aligned}\label{tau_w_v_1}
                    \xi_k^+&\in 
                    \begin{cases}
                    \xi_k \hspace{1.3cm}\text{if} \hspace{0.5cm} U(\bar{R}_k, \xi_k)-U(\bar{R}_k, \xi^*_k) \leq \delta_{\bar R}\\
                    \xi^*_k \hspace{1.3cm}\text{if} \hspace{0.5cm} U(\bar{R}_k, \xi_k)-U(\bar{R}_k, \xi^*_k) \geq \delta_{\bar R}
                    \end{cases}
                \end{aligned}
                }_{x \in \mathcal{J}_i}
\end{align}}In the feedback control law presented above, the angular velocity $\omega_i$, previously used in the second term of the proposed torque \eqref{R_obs_f}, is replaced with the new term $\psi\left(\tilde{Q}_i^\top  \nabla_{\tilde{Q}_i} U(\tilde{Q}_i, \zeta_i)\right)$. As already discussed, this term can be constructed using the outputs of the auxiliary system described in equations \eqref{q_11}-\eqref{q_21}. Furthermore, intuitively, the last term in the proposed torque \eqref{R_obs_f} can be replaced with:
\begin{equation}\label{aux_term}
\sum_{j\in\mathcal{N}_i}\left(\psi\left(\tilde{Q}_i^\top  \nabla_{\tilde{Q}_i} U(\tilde{Q}_i, \zeta_i)\right)-\psi\left(\tilde{Q}_j^\top  \nabla_{\tilde{Q}_j} U(\tilde{Q}_j, \zeta_j)\right)\right).
\end{equation}
However, incorporating this term in the feedback control law \eqref{tau_w_v}-\eqref{tau_w_v_1} introduces challenges in establishing the stability properties of the proposed velocity-free distributed hybrid attitude synchronization scheme, as no suitable Lyapunov function has been identified to proceed with the stability proof. Despite this theoretical limitation, simulations indicate that the scheme demonstrates convergence when the term in \eqref{aux_term} is considered. An alternative approach involves designing a hybrid auxiliary system, similar to the one described in \eqref{q_11}-\eqref{q_21}, to provide the required damping and compensate for the absence of relative angular velocity measurements (\ie, the last term in the proposed torque \eqref{R_obs_f}). However, this solution may increase the complexity of implementing the synchronization scheme and impose additional computational overhead.

Now, let $\hat x :=(\bar x, \tilde Q_1, \hdots, \tilde Q_N, \zeta_1, \hdots, \zeta_N) \in \hat{\mathcal{S}}:= \Bar{\mathcal{S}}\times SO(3)^N \times \mathbb{R}^N$. One can derive the following hybrid dynamics
{\small
\begin{equation}\label{hybrid_sys_w_v}
\hat{\mathcal{H}}:\begin{cases} {\dot{\hat x} = \hat F(\hat x),}&{\hat x \in \hat{\mathcal{F}}} \\ {{\hat x^ + } \in \hat G(\hat x),}&{\hat x \in \hat{\mathcal{J}}} \end{cases}
\end{equation}
}where
\begin{align}
    \hat{\mathcal{F}}&:=\{\hat x \in \hat{\mathcal{S}}: \bar x \in \bar{\mathcal{F}}~\text{and}~ \forall i \in \mathcal{V}, (\tilde Q_i, \zeta_i)\in \mathcal{F}_i^{\tilde Q}\}\nonumber\\
    \hat{\mathcal{J}}&:=\{\hat x \in \hat{\mathcal{S}}:  \bar x \in \bar{\mathcal{J}}~\text{or}~ \exists i \in \mathcal{V}, (\tilde Q_i, \zeta_i)\in \mathcal{J}_i^{\tilde Q}\}\nonumber
\end{align}
and $\hat G(\hat x_h) :=(\bar R_1, \hdots, \bar R_M, \{\xi_1, \xi_1^*\}, \hdots, \{\xi_M, \xi_M^*\}, \omega_1$ $, \hdots, \omega_N, \tilde Q_1, \hdots, \tilde Q_N, \{\zeta_1, \zeta_1^*\}, \hdots, \{\zeta_N, \zeta_N^*\})$,
{\small
\begin{equation}
    \hat F(\hat x) := \left[ {\begin{array}{c} {\bar R_1[\bar \omega_1]^{\times}} \\ \vdots\\ {\bar R_M[\bar \omega_M]^{\times}}\\{-k_\xi \nabla_{\xi_1}\bar U} \\ \vdots\\ {-k_\xi \nabla_{\xi_M}\bar U}\\{J_1^{-1} \left(-[\omega_1]^\times J_1 \omega_1+\tau_1\right)} \\ \vdots\\ {J_N^{-1} \left(-[\omega_N]^\times J_N \omega_N+\tau_N\right)}\\{\tilde{Q}_1 \left[\omega_1-k_Q \psi\left(\Tilde{Q}_1^\top  \nabla_{\Tilde{Q}_1}U(\Tilde{Q}_1, \zeta_1)\right)\right]^\times} \\ \vdots\\ {\tilde{Q}_N \left[\omega_N-k_Q \psi\left(\Tilde{Q}_N^\top  \nabla_{\Tilde{Q}_N}U(\Tilde{Q}_N, \zeta_N)\right)\right]^\times}\\{k_\zeta \nabla_{\zeta_1}U(\tilde Q_1, \zeta_1)}\\ \vdots\\{k_\zeta \nabla_{\zeta_N}U(\tilde Q_N, \zeta_N)} \end{array}}\right].\nonumber
\end{equation}}with $\tau_i$ is given in \eqref{tau_w_v} for each $i \in \mathcal{V}$. It follows from \eqref{hybrid_sys_w_v} that $\hat{\mathcal{F}}\cup \hat{\mathcal{J}} = \hat{\mathcal{S}}$. In addition, $\hat{\mathcal{F}}$ and $\hat{\mathcal{J}}$ are closed sets, and the hybrid closed-loop system \eqref{hybrid_sys_w_v} is autonomous and satisfies the hybrid basic conditions \cite[Assumption 6.5]{goebel2012hybrid}. Our third main result is stated in the following theorem:
\begin{thm}\label{theorem2}
    Let $k_{\tilde Q}, k_Q, k_\zeta, k_R, k_\xi >0$ and suppose Conditions \ref{hybrid_ass} and \ref{hybrid_ass_2} are satisfied. Then, the set $\hat{\mathcal{A}}:=\{ \hat x \in \hat{\mathcal{S}}:~\bar x \in \bar{\mathcal{A}}, \forall i \in \mathcal{V}, (\tilde Q_i, \zeta_i)=(I_3, 0)\}$ is globally asymptotically stable for the multi-agent hybrid closed-loop system (\ref{hybrid_sys_w_v}) and the number of jumps is finite.
\end{thm}
\begin{proof}
    See Appendix \ref{app_3}
\end{proof}

Similar to the design of \eqref{R_obs_f}-\eqref{R_obs_j}, the velocity-free distributed hybrid feedback control law \eqref{tau_w_v}-\eqref{tau_w_v_1} was developed based on generic potential functions. Considering again the potential function given in \eqref{expl_pf}, one can derive the explicit form of the distributed hybrid velocity-free feedback control law \eqref{tau_w_v}-\eqref{tau_w_v_1} as follows:

{\small
\begin{align}
             &\underbrace{
             \begin{aligned}\label{explicit_21}
                    \tau_i &=-k_R\bigg(\sum_{j \in \mathcal{O}_i} \psi\left(A \mathcal{R}_a(\xi_n, u)^\top  R_j^\top  R_i\right)+ \sum_{j \in \mathcal{I}_i} \mathcal{R}_a(\xi_p, u)\\
                    &\psi\left(A R_j^\top  R_i \mathcal{R}_a(\xi_p, u)\right)\hspace{-0.15cm}\bigg)-k_{\tilde Q} \mathcal{R}_a(\zeta_i, u) \psi\left(A Q_i^\top  R_i \mathcal{R}_a(\zeta_i, u)\right) \\
                    \dot{\xi}_k &=-k_\xi \left(\gamma \xi_k+2 u^\top  \psi\left(A \bar R_k \mathcal{R}_a(\xi_k, u)\right)\right)
                 \end{aligned}
                }_{x \in \mathcal{F}_i}\\
             &\underbrace{
             \begin{aligned}\label{explicit_22}
                    \xi_k^+&\in 
                    \begin{cases}
                    \xi_k \hspace{1.3cm}\text{if} \hspace{0.5cm} U(\bar{R}_k, \xi_k)-U(\bar{R}_k, \xi^*_k) \leq \delta_{\bar R}\\
                    \xi^*_k \hspace{1.3cm}\text{if} \hspace{0.5cm} U(\bar{R}_k, \xi_k)-U(\bar{R}_k, \xi^*_k) \geq \delta_{\bar R}
                    \end{cases}
                \end{aligned}
                }_{x \in \mathcal{J}_i}
\end{align}
}for every $i \in \mathcal{V}$ and $ k \in \mathcal{M}_i^+$. In addition, the hybrid dynamics of the auxiliary state $(Q_i, \zeta_i)$ are also given explicitly as follows:
{\small
\begin{align}
             &\underbrace{
                  \begin{aligned}\label{q_1}
                    \dot{Q}_i =&k_Q Q_i \left[Q^\top_i R_i \mathcal{R}_a(\zeta_i, u) \psi\left(A Q_i^\top  R_i \mathcal{R}_a(\zeta_i, u)\right)\right]^\times\\
                    \dot{\zeta}_i =&-k_\zeta \left(\gamma \zeta_i+2 u^\top  \psi\left(A Q^\top_i R_i \mathcal{R}_a(\zeta_i, u)\right)\right)
                 \end{aligned}
                }_{(Q_i, \zeta_i) \in \mathcal{F}_i^{\tilde Q}}\\
             &\underbrace{
                \begin{aligned}\label{q_2}
                    Q_i^+&=Q_i\\
                    \zeta_i^+&\in 
                    \begin{cases}
                    \zeta_i \hspace{1.3cm}\text{if} \hspace{0.5cm} U(\Tilde{Q}_i, \zeta_i)-U(\Tilde{Q}_i, \zeta^*_i) \leq \delta_{\tilde Q}\\
                    \zeta^*_i \hspace{1.3cm}\text{if} \hspace{0.5cm} U(\Tilde{Q}_i, \zeta_i)-U(\Tilde{Q}_i, \zeta^*_i) \geq \delta_{\tilde Q}
                    \end{cases}
                \end{aligned}
                }_{(Q_i, \zeta_i) \in \mathcal{J}_i^{\tilde Q}}
\end{align}
}where $i \in \mathcal{V}$. For the practical implementation of the velocity-free distributed hybrid feedback control law \eqref{explicit_21}-\eqref{explicit_22}, each agent will execute the dynamics of its corresponding auxiliary system. 
Note that while the feedback control scheme \eqref{explicit_21}-\eqref{explicit_22} eliminates the need for individual angular velocity measurements, it still requires each agent's orientation measurements. In contrast, the feedback control scheme \eqref{explicit_11}-\eqref{explicit_12} relies solely on relative orientations and individual angular velocities, avoiding the need for absolute orientation measurements. Unfortunately, eliminating angular velocity measurements in \eqref{explicit_21}-\eqref{explicit_22} comes at the cost of requiring absolute orientation measurements.

\section{SIMULATION}\label{s8}
In this section, we provide some numerical simulation results to illustrate the performance of the two distributed hybrid feedback control laws \eqref{explicit_11}-\eqref{explicit_12} and \eqref{explicit_21}-\eqref{explicit_22}, referred to as \textit{Hybrid Controller} and \textit{ Velocity-free Hybrid Controller}, respectively. Additionally, we include numerical simulations for the continuous feedback control law given in \eqref{continuous_tau}, referred to as \textit{ Continuous Controller}. We consider a network of seven satellites, where the satellites interact with each other according to the undirected graph topology depicted in Figure \ref{graph}. The neighbor sets are given as $\mathcal{N}_1=\{2\}$, $\mathcal{N}_2 = \{1, 3\}$, $\mathcal{N}_3 = \{2, 4, 6\}$, $\mathcal{N}_4 = \{3, 5\}$, $\mathcal{N}_5 = \{4\}$, $\mathcal{N}_6 = \{3, 7\}$ and $\mathcal{N}_7 = \{6\}$.
\begin{figure}[H]
    \centering
    \includegraphics[width=0.7\linewidth]{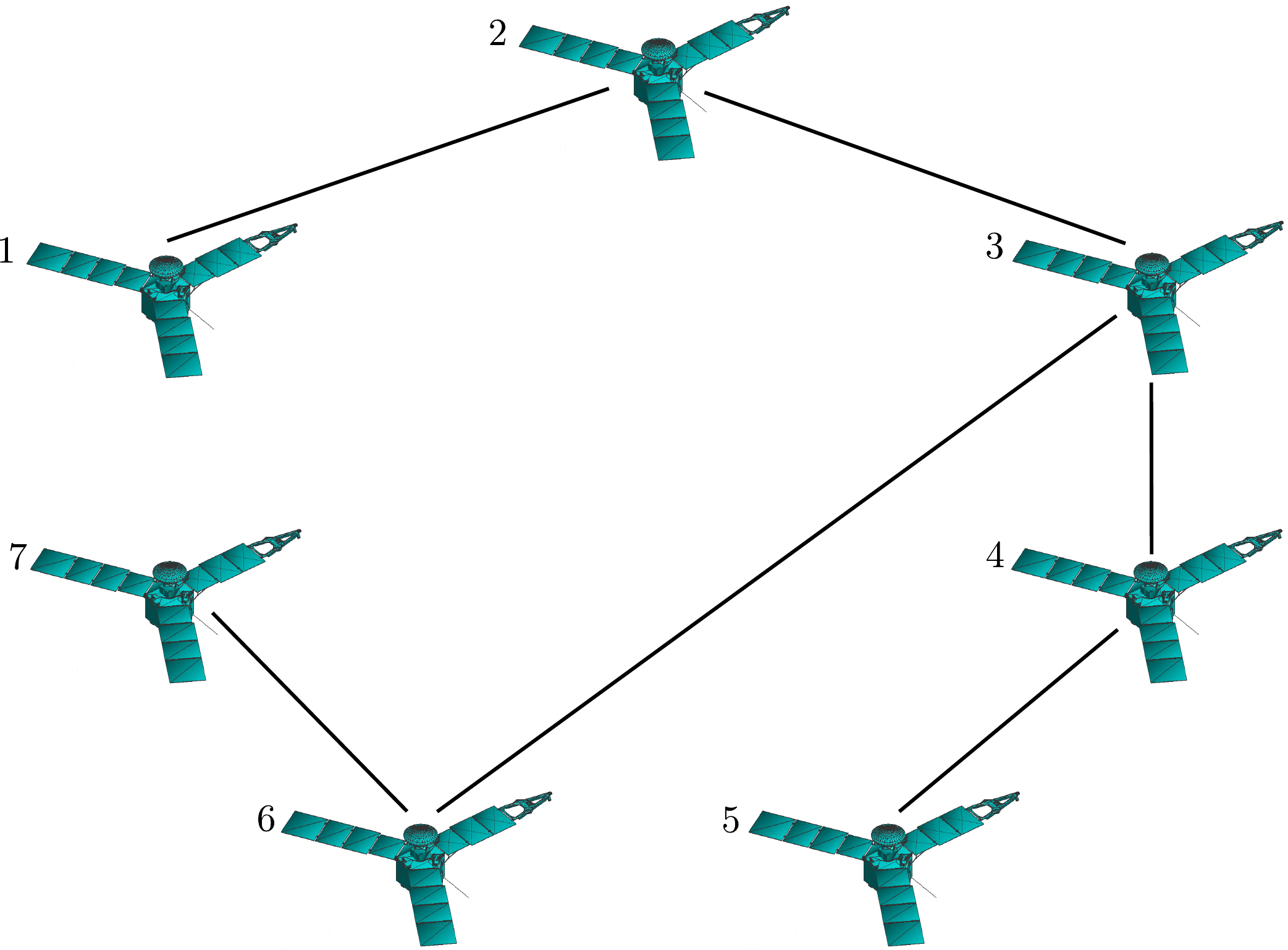}
    \caption{A network of seven satellites with an interaction graph.}
    \label{graph}
\end{figure}
We assign an arbitrary orientation to the interaction graph given in Figure \ref{graph}, and we index each oriented edge with a number as shown in Figure \ref{oriented_graph}. We consider the following initial conditions: $\omega(0)=0$, $\xi(0) = 0$, $\zeta(0)=0$, $R_1(0)=\mathcal{R}(-\frac{\pi}{2},\bar{u})$, $R_2(0)=\mathcal{R}(\frac{\pi}{2},\bar{u})$, $R_3(0)=\mathcal{R}(-\frac{\pi}{2},\bar{u})$, $R_4(0)=\mathcal{R}(\frac{\pi}{2},\bar{u})$, $R_5(0)=\mathcal{R}(-\frac{\pi}{2},\bar{u})$, $R_6(0)=\mathcal{R}(\frac{\pi}{2},\bar{u})$, $R_7(0)=\mathcal{R}(-\frac{\pi}{2},\bar{u})$, $Q_1(0)=\mathcal{R}(\frac{\pi}{2},\bar{u})$, $Q_2(0)=\mathcal{R}(-\frac{\pi}{2},\bar{u})$, $Q_3(0)=\mathcal{R}(\frac{\pi}{2},\bar{u})$, $Q_4(0)=\mathcal{R}(-\frac{\pi}{2},\bar{u})$, $Q_5(0)=\mathcal{R}(\frac{\pi}{2},\bar{u})$, $Q_6(0)=\mathcal{R}(-\frac{\pi}{2},\bar{u})$ and $Q_7(0)=\mathcal{R}(\frac{\pi}{2},\bar{u})$, with $\bar{u}=[0~0~1]^\top $. Note that these initial conditions are chosen such that the state belongs to the set of undesired equilibria. In addition, the gains and hybrid scheme parameters are set to $k_R=1$, $k_\omega=\bar k_\omega=0.1$, $k_Q=20$, $k_{\tilde Q}=2$, $k_\xi=k_\zeta=20$, $\delta_{\bar R}=\delta_{\tilde Q}=0.3848$, $\gamma=1.9251$, $\Xi=\Pi=\{0.9 \pi\}$, $u=[0~0.6455~0.7638]^\top $ and $A=\text{diag}([5, 8.57, 12])$. To simulate the \textit{Hybrid Controller} and the \textit{ Velocity-free Hybrid Controller}, we used the HyEQ Toolbox \cite{Sanfelice_matlab}.

\begin{figure}[H]
    \centering
    \includegraphics[width=0.7\linewidth]{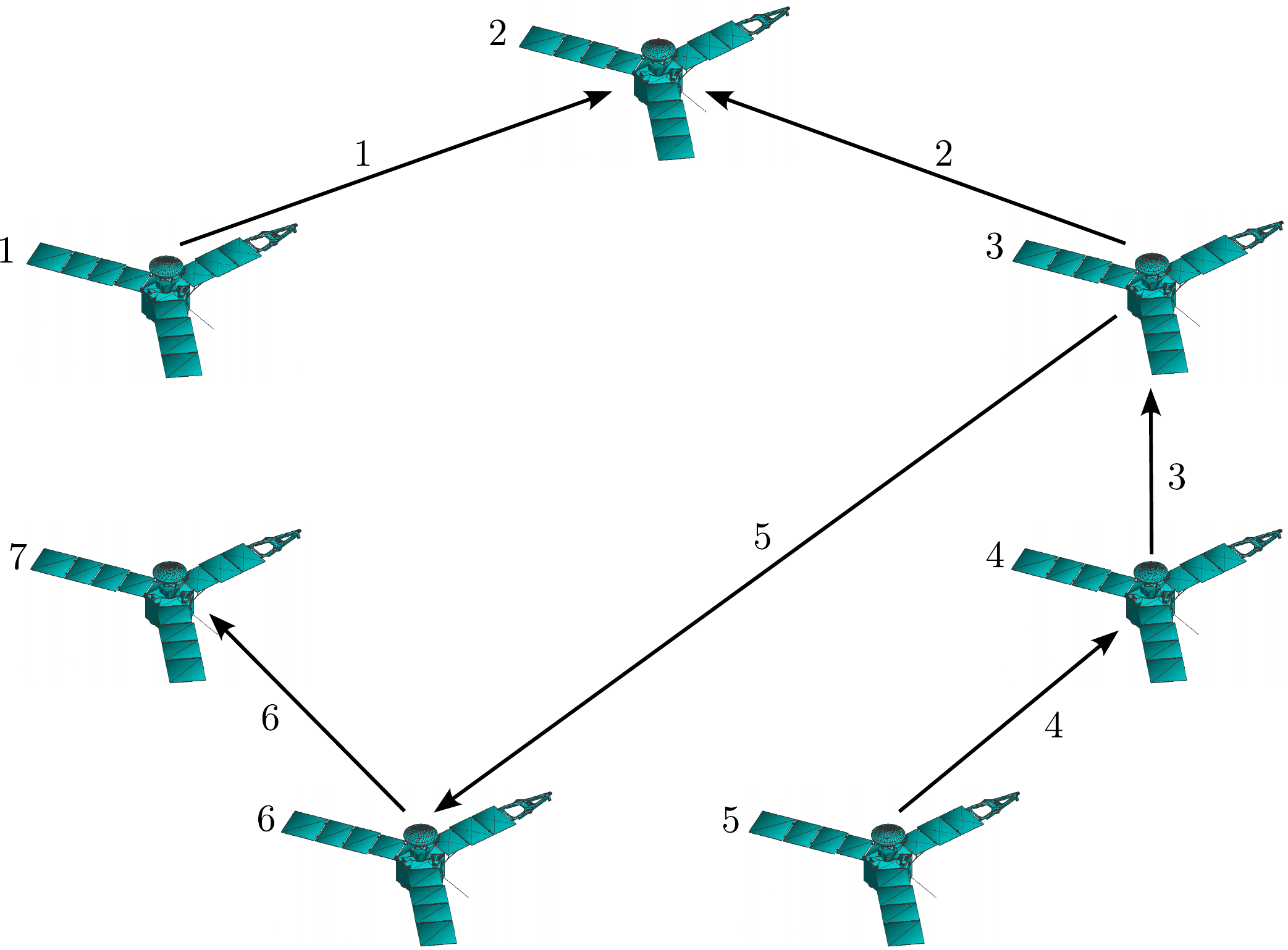}
    \caption{A network of seven satellites with an oriented interaction graph.}
    \label{oriented_graph}
\end{figure}
Figures \ref{R_bar}–\ref{xi} present the simulation results for the two distributed hybrid feedback control laws \eqref{explicit_11}-\eqref{explicit_12} and \eqref{explicit_21}-\eqref{explicit_22}, as well as for the continuous feedback scheme \eqref{continuous_tau}. As shown in Figure \ref{R_bar}, the two hybrid schemes, \ie, \eqref{explicit_11}-\eqref{explicit_12} and \eqref{explicit_21}-\eqref{explicit_22}, achieve faster convergence compared to the continuous controller \eqref{continuous_tau}. Ideally, under the continuous scheme where the initial states lie in the undesired equilibrium set, the agents' attitudes should remain at these undesired equilibria. However, due to numerical artifacts such as quantization or rounding errors inherent in MATLAB simulations\footnote{This phenomenon can also be interpreted as the effect of small perturbations or noise.}, the agents eventually escape from the undesired equilibria. As a result, the continuous scheme remains temporarily stuck before beginning to converge to the synchronized configuration. In contrast, the hybrid schemes avoid this issue through the switching mechanism associated with the variables $\xi_k$, which triggers an initial discrete jump from $0$ to $0.9\pi$. This jump allows the agents’ attitudes to escape the undesired equilibria from the start, leading to faster convergence. This behavior is also evident in Figure \ref{w}, where $\omega_i$, for every $i \in \mathcal{V}$, shows no response under the continuous scheme until approximately $t \approx 1.5\,s$, whereas the two hybrid schemes exhibit a response starting at $t = 0\,s$. Furthermore, Figures \ref{xi_1} and \ref{xi} illustrate the trajectories of the hybrid variables $\xi_k$ (for both hybrid controllers), for every $k \in \mathcal{M}$. These variables initially perform a discrete jump from $0$ to $0.9\pi$ (since the initial conditions lie in the jump set) and then converge smoothly to zero according to the vector fields given in \eqref{explicit_11} and \eqref{explicit_21}. Simulation videos are available at \href{https://youtu.be/VViBrnqnkns}{\textcolor{Rhodamine}{https://youtu.be/VViBrnqnkns}}.
\begin{figure}[H]
    \centering
    \includegraphics[width=1\linewidth,height=4.5cm]{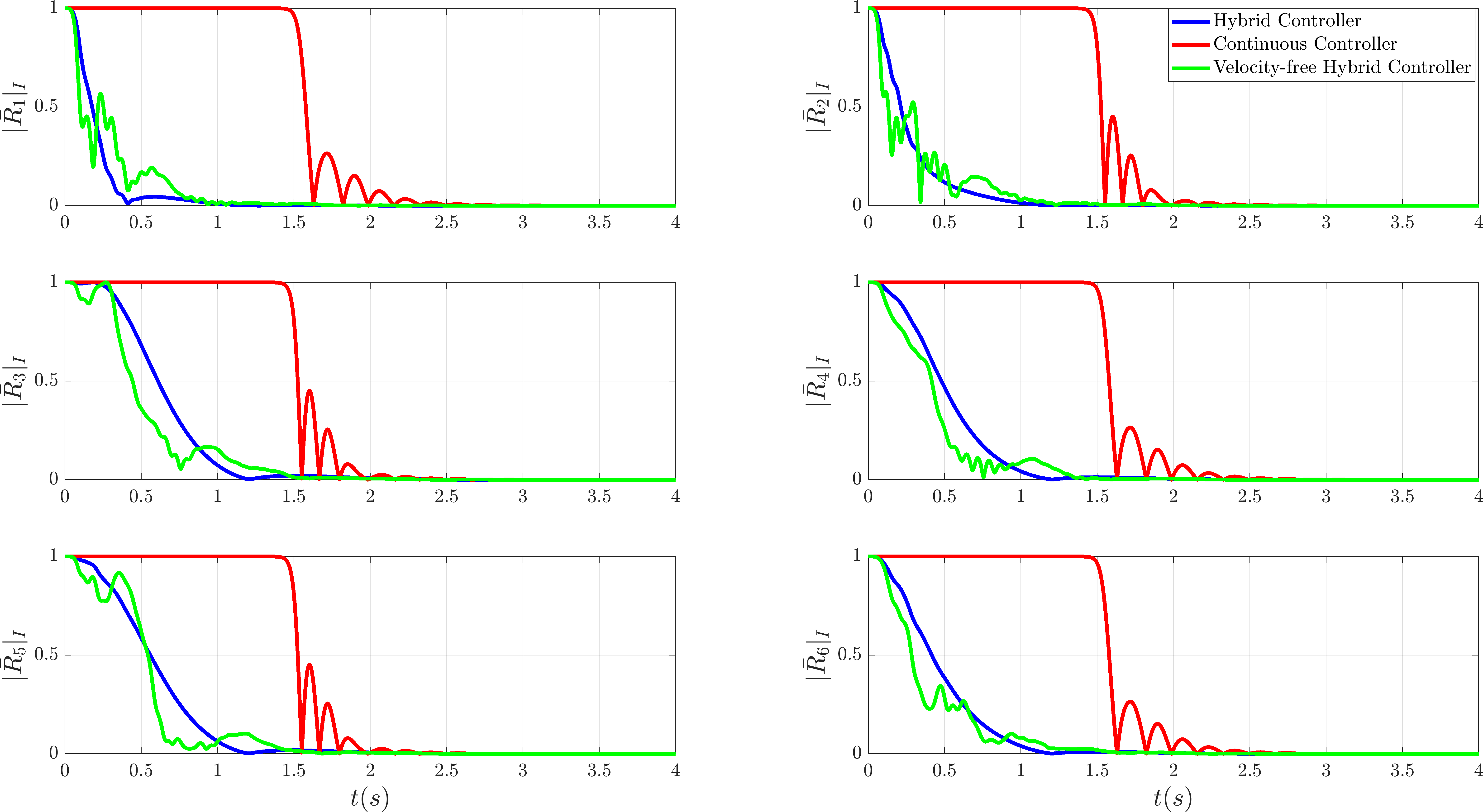}
    \caption{The time evolution of the relative attitude associated with each edge.}
    \label{R_bar}
\end{figure}

\begin{figure}[H]
    \centering
    \includegraphics[width=1\linewidth]{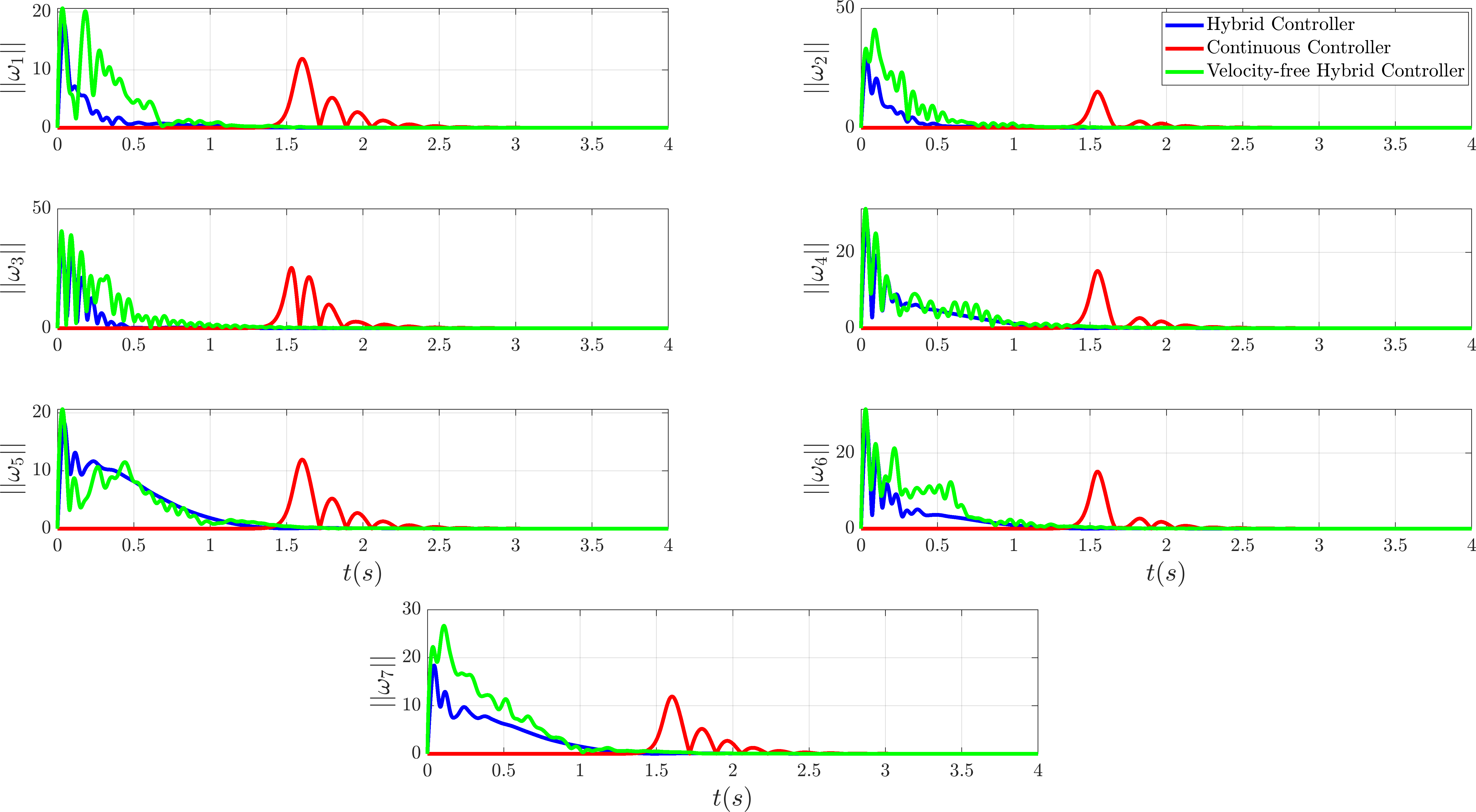}
    \caption{The time evolution of the angular velocity of each agent.}
    \label{w}
\end{figure}

\begin{figure}[h]
    \centering
    \includegraphics[width=0.7\linewidth]{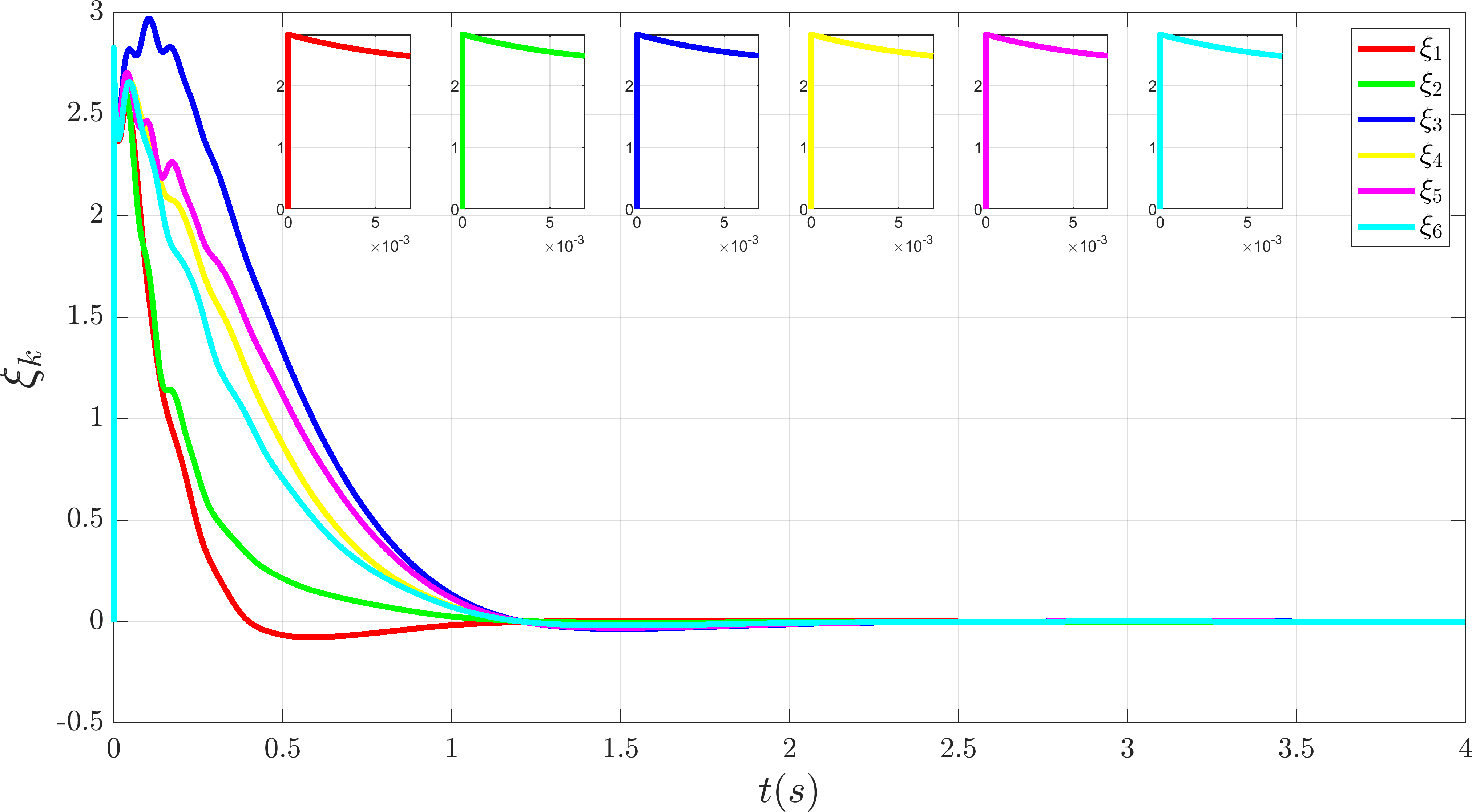}
    \caption{The time evolution of the auxiliary variable $\xi_k$ associated with each edge considering \textit{Hybrid Controller}.}
    \label{xi_1}
\end{figure}

\begin{figure}[H]
    \centering
    \includegraphics[width=0.7\linewidth]{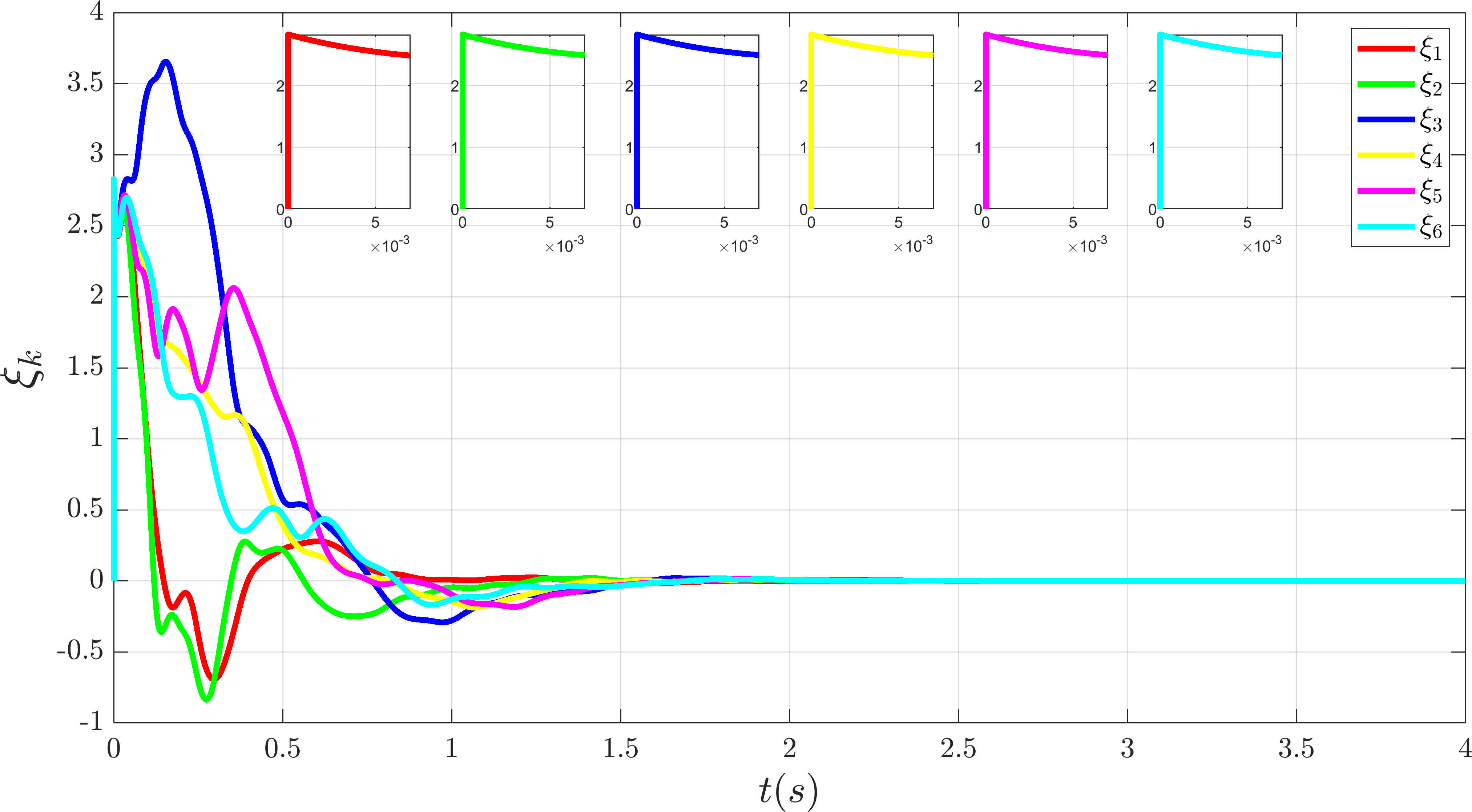}
    \caption{The time evolution of the auxiliary variable $\xi_k$ associated with each edge \textit{Velocity-free Hybrid Controller}.}
    \label{xi}
\end{figure}

In the second simulation, we consider the continuous feedback scheme \eqref{continuous_tau} and the hybrid one \eqref{explicit_11}-\eqref{explicit_12} with $k_\omega = 0$ and $\bar{k}_\omega > 0$. As highlighted in Remark \ref{rmk_tv_ori} and shown in Figures \ref{R_contin_k_w_0} and \ref{R_hybrid_k_w_0}, the absolute attitude of each agent under both feedback laws \eqref{continuous_tau} and \eqref{explicit_11}-\eqref{explicit_12}, with $k_\omega = 0$ and $\bar{k}_\omega > 0$, converges to a common time-varying orientation.

\begin{figure}[H]
    \centering
    \includegraphics[width=0.6\linewidth]{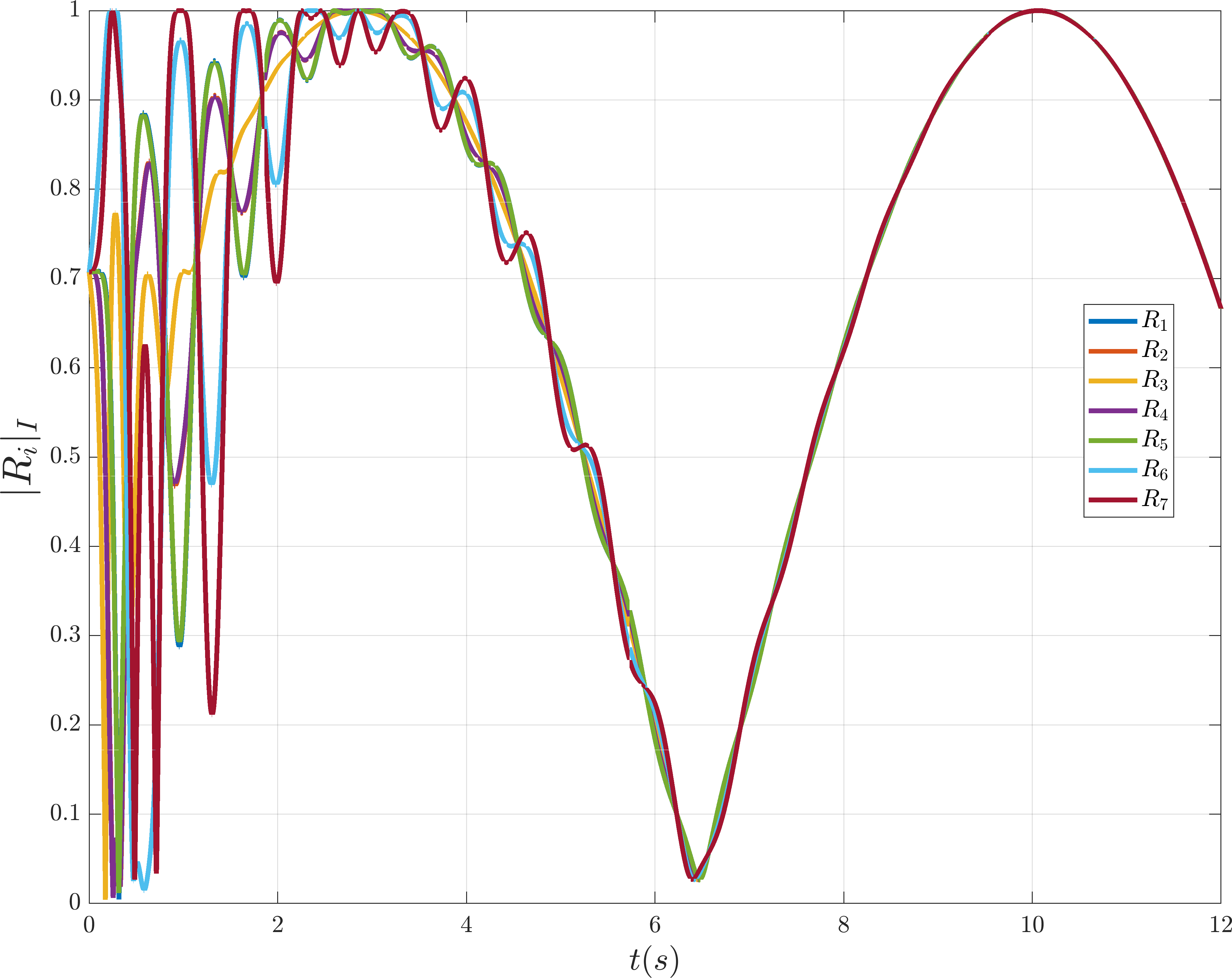}
    \caption{The time evolution of the absolute attitude norm $|R_i|_I$ for each agent under the control law \eqref{continuous_tau} with $k_\omega = 0$ and $\bar{k}_\omega > 0$.}
    \label{R_contin_k_w_0}
\end{figure}

\begin{figure}[H]
    \centering
    \includegraphics[width=0.6\linewidth]{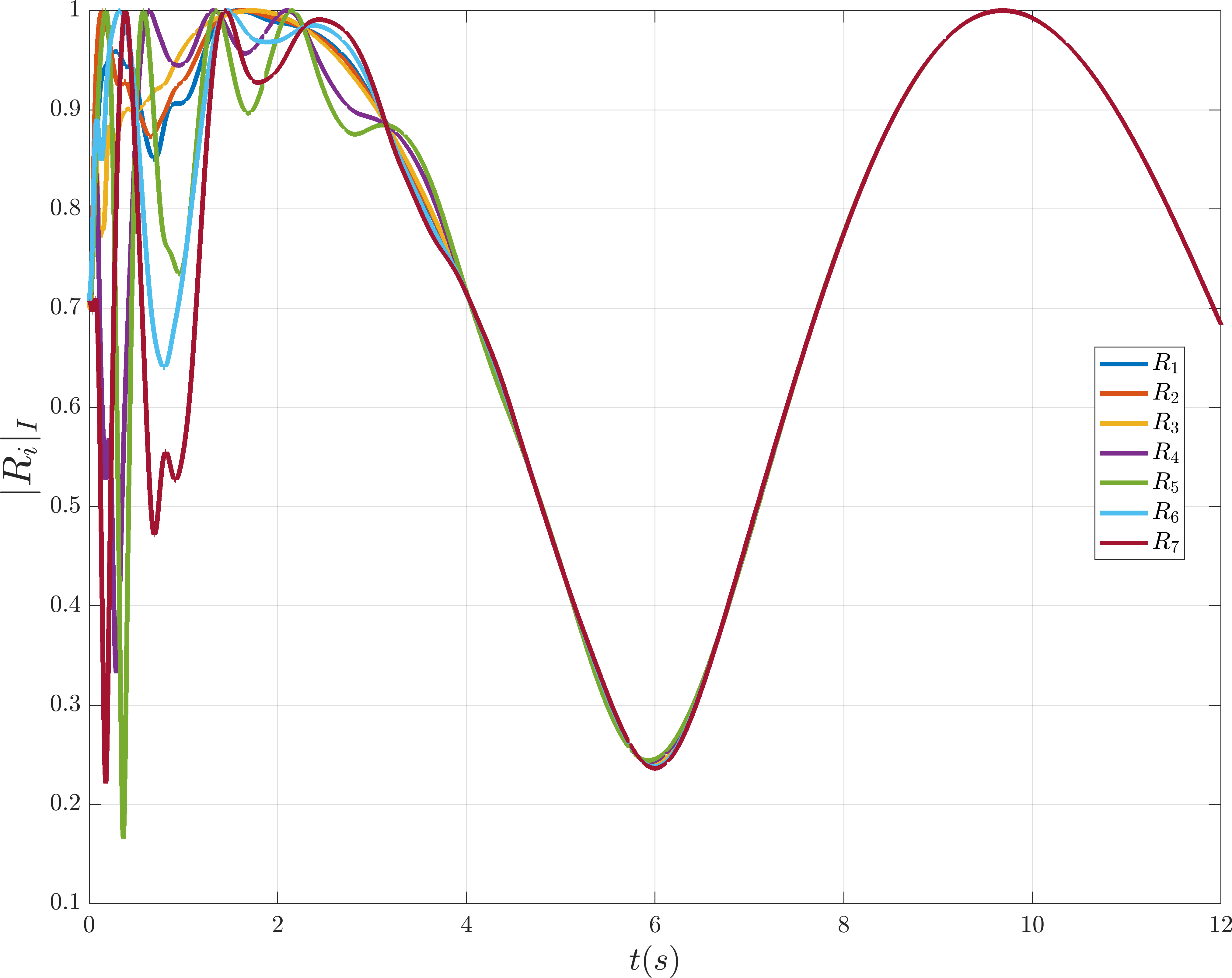}
    \caption{The time evolution of the absolute attitude norm $|R_i|_I$ for each agent under the control law \eqref{explicit_11} with $k_\omega = 0$ and $\bar{k}_\omega > 0$.}
    \label{R_hybrid_k_w_0}
\end{figure}

\section{CONCLUSIONS}\label{s9}
We have presented three attitude synchronization schemes for rigid body systems on $SO(3)$ under undirected, connected, and acyclic communication graph topology. The first continuous distributed attitude synchronization scheme is quite similar to the scheme proposed in \cite{SARLETTE2009572}, where only local asymptotic stability was claimed. We provided a rigorous proof showing that this scheme has almost global asymptotic stability, which is the strongest possible result with smooth control inputs. This result stems from the existence of a finite set of unstable (undesired) equilibria with a stable manifold of zero Lebesgue measure. To overcome this problem and achieve global asymptotic synchronization, we proposed a new distributed hybrid feedback control scheme on $SO(3)$, relying on angular velocity measurements and relative attitude information, guaranteeing global convergence of the individual orientations to a common orientation. Furthermore, a velocity-free distributed hybrid attitude synchronization scheme on $SO(3)$, with global asymptotic stability guarantees, relying only on attitude measurements, has been proposed. The proposed velocity-free control law uses an auxiliary dynamical system for each agent to generate the necessary damping that compensates for the lack of angular velocity information. Note that both proposed hybrid distributed attitude synchronization schemes have been designed under the assumption that the inter-agent interaction topology is fixed with no communication time delays. Unfortunately, this is not the case in many real-world multi-agent application scenarios.  Therefore, redesigning our proposed schemes for multi-agent rigid body systems under dynamically changing and delayed inter-agent communication topology is an interesting future work.

\appendices 
\section{Proof of Theorem \ref{theorem_continuous}} \label{app_1}
Considering the undirected graph $\mathcal{G}$ with a virtual (arbitrary) orientation, one has $\mathcal{N}_i=\mathcal{I}_i\cup\mathcal{O}_i$, with $\mathcal{I}_i:=\{j\in \mathcal{N}_i: j \hspace{0.1cm} \text{is the tail of the oriented edge} \hspace{0.1cm} (i,j) \in \mathcal{E}\}$ and $\mathcal{O}_i:=\{j\in \mathcal{N}_i: j \hspace{0.1cm} \text{is the head of the oriented edge} \hspace{0.1cm} (i,j) \in \mathcal{E}\}$. From this, one can rewrite the first term of \eqref{continuous_tau} as follows:

{\small
\begin{align}
    \sum_{j \in \mathcal{N}_i}& \psi(A R_j^\top R_i) =\sum_{j\in \mathcal{I}_i}\psi(AR_j^\top R_i)+\sum_{j\in \mathcal{O}_i}\psi(AR_j^\top R_i)\nonumber\\
    &= \sum_{j\in \mathcal{I}_i}\psi(AR_j^\top R_i)-\sum_{j\in \mathcal{O}_i}\psi(R_i^\top R_j A)\label{eq_1}\\
    &= \sum_{j\in \mathcal{I}_i}\psi(AR_j^\top R_i)-\sum_{j\in \mathcal{O}_i}R_i^\top R_j\psi(AR_i^\top R_j)\label{eq_2}\\
    &=\nonumber\sum_{n\in \mathcal{M}_i^+}\psi(A\bar{R}_n)-\sum_{l\in \mathcal{M}_i^-}\bar{R}_l\psi(A\bar{R}_l)= \sum_{k=1}^M H_{ik}\psi(A\bar{R}_k),
\end{align}
}where $H_{ik}$ is given in (\ref{H_bar}). Equations (\ref{eq_1}) and (\ref{eq_2}) are obtained using the facts that $\psi(BR)=-\psi(R^{\top}B)$ and $\psi(GR)=R^{\top}\psi(RG)$, $\forall G, B=B^{\top}\in \mathbb{R}^{3\times3} $ and $R\in SO(3)$. Let $\tau:=[\tau_1^\top, \tau_2^\top, \hdots, \tau_N^\top]^\top \in \mathbb{R}^{3N}$. One can derive the following stacked form of control torque \eqref{continuous_tau}:
\begin{equation}\label{compact_continuous_tau} 
    \tau = - k_R \bar H \Psi -k_\omega \omega -\bar k_\omega (\mathcal{L}\otimes I_3) \omega,
\end{equation}
where $\Psi:=\left[\psi(A\bar R_1)^\top, \psi(A\bar R_2)^\top, \hdots, \psi(A\bar R_M)^\top\right]^\top \in \mathbb{R}^{3M}$ and $\mathcal{L}:= H H^\top  \in \mathbb{R}^{N\times N}$ is the Laplacian matrix corresponding to the graph $\mathcal{G}$ and the matrix $H$ is given in \eqref{h_matrix}. Furthermore, the stacked form of the angular velocity dynamics can be written as follows:
\begin{equation}\label{compact_omega_dynamics} 
    \mathbf{J} \dot \omega =  - \mathbf{W}^\times\mathbf{J} \omega- k_R \bar H \Psi -k_\omega \omega - \bar k_\omega (\mathcal{L}\otimes I_3) \omega,
\end{equation}
where $\mathbf{W}^\times=:\text{diag}([\omega_1]^\times, [\omega_2]^\times, \hdots, [\omega_N]^\times)\in \mathbb{R}^{3N\times 3N}$ and $\mathbf{J} :=\text{diag}(J_1, J_2, \hdots, J_N)\in \mathbb{R}^{3N\times 3N}$. Now, let us proceed with the stability analysis for the dynamics \eqref{R_bar_dynamics_k}-\eqref{w_dynamics_k} with control torque \eqref{continuous_tau}. Consider the following Lyapunov function candidate:
\begin{equation}\label{pf}
    \mathcal{V}_z(z) = k_R\sum_{k=1}^{M} \text{tr}\left(A(I_3-\bar{R}_k)\right)+\omega^\top \mathbf{J} \omega.
\end{equation}
This Lyapunov function candidate is positive definite on $\mathcal{S}_z$ with respect to $\mathcal{A}_z$. The time-derivative of $ \mathcal{V}_z(z)$, along the trajectories of the closed-loop system \eqref{R_bar_dynamics_k} and \eqref{compact_omega_dynamics}, is given by
\begin{equation}\label{v_z_dot}
    \dot{\mathcal{V}}_z(z)=- 2 k_\omega ||\omega||^2 -2 \bar k_\omega ||(H^\top \otimes I_3) \omega||^2.
\end{equation}
The last equality was obtained using the facts that $\text{tr}\left(B[x]^\times\right)=\text{tr}\left(\mathbb{P}_a(B)[x]^\times\right)$, $\text{tr}\left([x]^\times [y]^\times\right)=-2x^{\top}y$, $\forall x, y\in \mathbb{R}^3$ and $\forall B \in \mathbb{R}^{3\times3}$ and $\omega^\top (\mathcal{L}\otimes I_3)\omega=||(H^\top \otimes I_3) \omega||^2$. The equality \eqref{v_z_dot} shows that the desired equilibrium set $\mathcal{A}_z$ is stable. Moreover, as per LaSalle’s invariance theorem, any solution $z$ to the dynamics \eqref{R_bar_dynamics_k}-\eqref{w_dynamics_k} with \eqref{continuous_tau}, $\forall i \in \mathcal{V}$, must converge to the largest invariant set contained in the set characterized by $\dot{\mathcal{V}}_z(z)=0$, \ie, $\omega=0$. From \eqref{compact_omega_dynamics}, it follows that $\bar{H} \Psi = 0$. Applying \cite[Lemma 2]{Mouaad_ACC23}, the last equality implies $\Psi = 0$, which further leads to  
\begin{equation} \label{eq_AR}  
     A \bar{R}_k = \bar{R}_k^{\top} A,    
\end{equation}  
for all \(k \in \mathcal{M}\). From \eqref{eq_AR}, we can deduce that \(\mathcal{A}_z \subset \Upsilon_z\). Furthermore, by employing similar arguments to those presented in \cite[Lemma 2]{Mayhew_ACC2011}, every solution \(z\) of the dynamics \eqref{R_bar_dynamics_k}-\eqref{w_dynamics_k} with \eqref{continuous_tau}, \(\forall i \in \mathcal{V}\), must converge to the set \(\Upsilon_z\). This concludes the proof of item \eqref{set_of_equilibrium}.

To prove items (\ref{unstability_of_equilibrium}) and (\ref{stability_of_equilibrium}), we first present the following lemma:
\begin{lem}\label{lem_cp_pf}
    Consider the potential function $\mathcal{V}_z(z)$ given in \eqref{pf}. The following statements hold:
    \begin{enumerate}[I)]
        \item The critical points of $\mathcal{V}_z(z)$ coincide with the set of equilibria $\Upsilon_z$. \label{itm_I1}
        \item The critical points of \( \mathcal{V}_z(z) \) within \( \Upsilon_z \setminus \mathcal{A}_z \) are saddle points.\label{itm_I2}
    \end{enumerate}
\end{lem}
\begin{proof}
    See Appendix \ref{proof:lem_cp_pf}
\end{proof}
 Now, we proceed with the proof of the instability of undesired equilibrium points $\Upsilon_z \setminus \mathcal{A}_z$ for the dynamics \eqref{R_bar_dynamics_k}-\eqref{w_dynamics_k} with the control torque \eqref{continuous_tau}. Consider the real-valued function $\bar{\mathcal{V}}_z(z): SO(3)^M \times \mathbb{R}^{3N}\rightarrow \mathbb{R}$, defined as follows:  
\begin{align}\label{V_bar_1}
    \bar{\mathcal{V}}_z(z) = 2\, k_R \sum_{n \in \mathcal{M}^\pi} (\lambda_{p_n} + \lambda_{d_n}) - \mathcal{V}_z(z),
\end{align}  
where $\lambda_{p_n}$ and $\lambda_{d_n}$ are two distinct eigenvalues of the matrix $A$, \ie, $p_n, d_n \in \{1, 2, 3\}$ with $p_n \neq d_n$. Let $z^* \in \Upsilon_z \setminus \mathcal{A}_z$ denote an undesired equilibrium point such that $\bar{R}_n = \mathcal{R}(\pi, u_{l_n})$ for $n \in \mathcal{M}^\pi$, where $l_n \in \{1, 2, 3\}$ satisfies $l_n \neq p_n$ and $l_n \neq d_n$. Clearly, $\bar{\mathcal{V}}_z(z^*) = 0$. Define the set $\mathbb{B}_r := \{ (\bar{R}_1, \bar{R}_2, \dots, \bar{R}_M, \omega_1, \omega_2, \dots, \omega_N) \in \mathcal{S}_z : |\bar{R}_1^\top \bar{R}_1^*|_I + |\bar{R}_2^\top \bar{R}_2^*|_I + \dots + |\bar{R}_M^\top \bar{R}_M^*|_I+||\omega_1||+||\omega_2||+\dots+||\omega_N|| \leq r \}$ with $r > 0$. Since the set \(\Upsilon_z \setminus \mathcal{A}_z\) consists only of the saddle critical points of \(\mathcal{V}_z(z)\) (as established in Lemma~\ref{lem_cp_pf}), the function \(\mathcal{V}_z(z)\) does not attain a strict local minimum or maximum at these points. Consequently, in the neighborhood of any \(z^* \in \Upsilon_z \setminus \mathcal{A}_z\), there exists at least one point \(\tilde z\) such that \(\mathcal{V}_z(\tilde z) < \mathcal{V}_z(z^*)\) and another point \(\hat z\) such that \(\mathcal{V}_z(\hat z) > \mathcal{V}_z(z^*)\). Therefore, for any \(z^* \in \Upsilon_z \setminus \mathcal{A}_z\), we can always find points in \(\mathbb{B}_r\) where \(\bar{\mathcal{V}}(z) = \mathcal{V}_z(z) - \mathcal{V}_z(z^*) > 0\). This implies that the set $\mathbb{U} = \{ z \in \mathbb{B}_r \mid \bar{\mathcal{V}}(z) > 0 \}$ is non-empty. Moreover, in view of \eqref{v_z_dot} and \eqref{V_bar_1}, one has that $\dot{\bar{\mathcal{V}}}(z) = -\dot{\mathcal{V}}(z) > 0$ in $\mathbb{U}$, which implies that any trajectory originating in the set $\mathbb{U}$ must exit $\mathbb{U}$ from boundary surface of $\mathbb{B}_r$. By virtue of \textit{Chetaev's theorem} \cite{khalil2002nonlinear}, it can be concluded that all points in the undesired equilibrium set $\Upsilon_z \setminus \mathcal{A}_z$ are unstable. This completes the proof of item (\ref{unstability_of_equilibrium}).

Now, to show the almost global asymptotic stability of the desired equilibrium set $\mathcal{A}_z$, we need to prove that the stable manifold associated with the undesired equilibrium set $\Upsilon_z \setminus \mathcal{A}_z$ has zero Lebesgue measure. To do so, we need to show that the linearization of the system \eqref{R_bar_dynamics_k}--\eqref{w_dynamics_k}, together with \eqref{continuous_tau}, at the undesired equilibrium set $\Upsilon_z \setminus \mathcal{A}_z$, leads to a Jacobian matrix with no eigenvalues on the imaginary axis, and admits at least one eigenvalue with a positive real part. Let $\bar R_k=\bar R_k^* \exp{\left([\bar r_k]^\times\right)}$ and $\omega_i = y_i$,  where $\bar r_k , y_i\in \mathbb{R}^3$ are sufficiently small and $\bar R_k^* \in \{\bar R_n, \bar R_m\}$. Considering the expression of $\bar R_k$ and using the first-order approximation $\exp\big([\bar r_k]^\times\big) \approx I_3 + [\bar r_k]^\times$ for sufficiently small $\bar r_k$, we obtain the following linearization of $\bar R_k$ and $\omega_i$ around any point in the undesired equilibrium set $\Upsilon_z \setminus \mathcal{A}_z$:
\begin{equation}\label{first_order_app}
    \bar R_k = \bar R_k^*(I_3 + [\bar r_k]^\times), \quad \omega_i = y_i,
\end{equation}
for every $k \in \mathcal{M}$ and $i \in \mathcal{V}$. Define $\bar r:=\left[\bar r_1^\top, \bar r_2^\top,\dots, \bar r_M^\top\right]^\top \in \mathbb{R}^{3M}$ and $y:=\left[y_1^\top, y_2^\top,\dots, y_N^\top\right]^\top \in \mathbb{R}^{3N}$, and let $\mathbf{A}^* = \text{diag}(A^*_1, A^*_2, \hdots, A^*_M)\in \mathbb{R}^{3M\times 3M}$, where $A^*_k=\text{tr}(A\bar R^*_k)I_3 - A \bar R^*_k$. It follows from the dynamics in \eqref{R_bar_dynamics_k}--\eqref{w_dynamics_k} together with \eqref{continuous_tau} that
\begin{align}\label{linear_sys_und}
\begin{bmatrix}
    \dot{\bar r} \\[4pt] \dot y
\end{bmatrix} =
\mathbf{M}
\begin{bmatrix}
    \bar r \\[4pt] y
\end{bmatrix},
\end{align}
where {\small $$\mathbf{M}:=\begin{pmatrix}
    0_{3M\times 3M} & \bar{H}_c^\top\\
    -\frac{k_R}{2} \mathbf{J}^{-1} \bar{H}_c \mathbf{A}^* & -\mathbf{J}^{-1} \left(k_\omega I_{3N}+\bar k_\omega (\mathcal{L}\otimes I_3)\right) 
\end{pmatrix}$$}with 
{\small \begin{equation}\label{H_bar}
\bar{H}_c := [\bar h^c_{ik}]_{N\times M}, \qquad 
\bar h^c_{ik} =
\begin{cases}
    I_3, & k \in \mathcal{M}_i^+,\\
    -\bar R_k^*, & k \in \mathcal{M}_i^-,\\
    0, & \text{otherwise}.
\end{cases}
\end{equation}
}Note that $\bar{H}_c$ is a constant matrix. Since the interaction graph is an undirected tree, and by following arguments similar to those in the proof of \cite[Lemma 2]{Mouaad_ACC23}, one can show that $\bar{H}_c$ has full column rank. The linear dynamics \eqref{linear_sys_und} has been obtained using the fact that $B [z]^\times + [z]^\times B^{\top} = \left[\left(\mathrm{tr}(B)I_3-B^{\top}\right) z\right]^\times$, $ \forall z\in \mathbb{R}^3$, $\forall B \in \mathbb{R}^{3\times3}$. Moreover, the determinant of the matrix $\mathbf{M}$ is given by
{\small
\begin{align}\label{det_M}
  \text{det}(\mathbf{M}) &= \frac{k_R}{2}\det(\mathbf{J}^{-1}) 
   \text{det}\left( k_\omega I_{3N} + \bar k_\omega (\mathcal{L}\otimes I_3)\right) \nonumber\\
  &\quad \times 
   \text{det}\left(\bar H_c^\top \bigl(k_\omega I_{3N} + \bar k_\omega (\mathcal{L}\otimes I_3)\bigr)^{-1}\bar H_c\right)
   \text{det}(\mathbf{A}^*) \neq 0.
\end{align}
}Hence, $\mathbf{M}$ does not have an eigenvalue at the origin. The inequality in \eqref{det_M} follows from the facts that $\mathbf{J}$ is positive definite, $\mathbf{A}^*$ is invertible (its eigenvalues are nonzero), $\bar H_c$ has full column rank, and $\mathcal{L}$ is positive semidefinite. Now, we show that $\mathbf{M}$ has no eigenvalues with zero real part. To do so, we proceed by contradiction. Assume that $\mathbf{M}$ has an imaginary eigenvalue $i\lambda$, where $i^2=-1$ and $\lambda \in \mathbb{R}\setminus\{0\}$, with a corresponding eigenvector
$
z=(z_1,z_2),
$
with $z_1\in\mathbb{R}^{3M},\; z_2\in\mathbb{R}^{3N}$
such that $\mathbf{M}z=i\lambda z$. It follows from \eqref{linear_sys_und} that 
\begin{align}
&\bar{H}_c^\top z_2 = i\lambda z_1, \label{sys_equ_1}\\
- \frac{k_R}{2} \mathbf{J}^{-1} \bar{H}_c \mathbf{A}^* z_1 -&\mathbf{J}^{-1} \left(k_\omega I_{3N}+\bar k_\omega (\mathcal{L}\otimes I_3)\right) z_2 = i\lambda z_2. \label{sys_equ_2}
\end{align}
Equation \eqref{sys_equ_1} can be rewritten as follows: 
\begin{equation}\label{sys_equ_3}
    z_1 = -\frac{i}{\lambda} \bar{H}_c^\top z_2.
\end{equation}
From \eqref{sys_equ_2} and \eqref{sys_equ_3}, one has 
\begin{equation}\label{sys_equ_4}
    \frac{ik_R}{2\lambda} \mathbf{J}^{-1} \bar{H}_c \mathbf{A}^*\bar H_c^\top z_2 -\mathbf{J}^{-1} \left(k_\omega I_{3N}+\bar k_\omega (\mathcal{L}\otimes I_3)\right) z_2 = i\lambda z_2.
\end{equation}
Furthermore, by multiplying \eqref{sys_equ_4} on the left by $z_2^\top \mathbf{J}$ and rearranging it, one obtains:
\begin{align}\label{sys_equ_5}
    &\left(\frac{k_R}{2\lambda} z_2^\top\bar{H}_c \mathbf{A}^*\bar H_c^\top z_2-\lambda z_2^\top \mathbf{J} z_2 \right)i\nonumber\\
    &~~~~~~~~~~~~~~~~~~~~~~~-z_2^\top \left(k_\omega I_{3N}+\bar k_\omega (\mathcal{L}\otimes I_3)\right) z_2 =0.
\end{align}
It follows from the real part of \eqref{sys_equ_5} that
\begin{align}
    k_\omega ||z_2||^2 +\bar k_\omega ||(H^\top\otimes I_3) z_2||^2 =0,
\end{align}
which implies that $z_2 = 0$. This further implies that $z_1 = 0$ according to \eqref{sys_equ_3}. Hence, the corresponding eigenvector is the zero vector, which is a contradiction. Therefore, the matrix $\mathbf{M}$ does not have any eigenvalues with zero real part. Together with the fact that all points in the undesired equilibrium set $\Upsilon_z \setminus \mathcal{A}_z$ are unstable and the linearized system \eqref{linear_sys_und} does not have an eigenvalue at the origin, this implies that $\mathbf{M}$ must have at least one eigenvalue with a positive real part. Thus, by the stable manifold theorem \cite{Perko_book} and the fact that the vector field defined by the dynamics \eqref{R_bar_dynamics_k}–\eqref{w_dynamics_k} under the continuous control torque \eqref{continuous_tau} is at least $C^1$, the stable manifold associated with the undesired equilibrium set $\Upsilon_z \setminus \mathcal{A}_z$ has zero Lebesgue measure. Consequently, the equilibrium set $\mathcal{A}_z$ is almost globally asymptotically stable. This completes the proof of item (\ref{stability_of_equilibrium}).

\section{Proof of Lemma \ref{lem_cp_pf}}\label{proof:lem_cp_pf}
To verify that the set of critical points of the potential function \eqref{pf} coincides with the set of equilibrium points of the dynamics described by \eqref{R_bar_dynamics_k}-\eqref{w_dynamics_k} under the control torque given in \eqref{continuous_tau}, we will compute the gradient of the potential function $\mathcal{V}_z$ with respect to $\bar{R}_k$ and $\omega_i$ for all $k \in \mathcal{M}$ and $i \in \mathcal{V}$. Let $\mathbb{O} \subset \mathbb{R}$ be an open interval containing zero in its interior. For each $k \in \mathcal{M}$ and $i \in \mathcal{V}$, we define the two smooth curves $\varphi_k: \mathbb{O} \to SO(3)$ and $\gamma_i: \mathbb{O} \to \mathbb{R}^3$ as follows: 
\begin{align}
    \varphi_k(t) = \bar{R}_k \exp\left(t [\zeta_k]^\times\right), \quad
    \gamma_i(t) = \omega_i+v_i t,\label{curve_2}
\end{align}
where $\bar R_k \in SO(3)$, $\omega_i\in \mathbb{R}^3$, $ \zeta_k\in \mathbb{R}^3$ and $ v_i \in \mathbb{R}^3$ for every $k \in \mathcal{M}$ and $i \in \mathcal{V}$. Define $\bar z(t) := \left(\varphi_1(t), \varphi_2(t), \ldots, \varphi_M(t), \gamma_1(t), \gamma_2(t), \ldots, \gamma_N(t)\right) \in \mathcal{S}_z$. The derivative of $\mathcal{V}_z\left(\bar z (t)\right)$ with respect to $t$ is given by:  
\begin{align}\label{hess}
    \frac{d}{dt}\mathcal{V}_z\left(\bar z (t)\right)=&-k_R\sum_{k=1}^{M} \text{tr}\left(A\bar{R}_k\exp{\left(t [\zeta_k]^\times\right)}[\zeta_k]^\times\right)\nonumber\\
    &+2\sum_{i=1}^N v_i^\top J_i (\omega_i+v_i t).
\end{align}
Moreover, at $t=0$, one has
\begin{align}\label{gradient_v}
    &\left.\frac{d}{dt}\mathcal{V}_z\left(\bar z (t)\right)\right|_{t=0}=-k_R\sum_{k=1}^{M} \text{tr}\left(A\bar{R}_k[\zeta_k]^\times\right)+2\sum_{i=1}^N v_i^\top J_i \omega_i\nonumber\\
    =&2\,k_R\sum_{k=1}^{M} \zeta_k^\top \psi(A \bar R_k)+2\sum_{i=1}^N v_i^\top J_i \omega_i
    =2 \begin{bmatrix}
        \zeta\\v
    \end{bmatrix}^\top
    \begin{bmatrix}
        k_R\,\Psi\\ \mathbf{J} \omega
    \end{bmatrix}.
\end{align}
Note that {\small $\mathbf{J}\omega=\left[\left(\nabla_{\omega_1} \mathcal{V}_z\right)^\top,\left(\nabla_{\omega_2} \mathcal{V}_z\right)^\top, \hdots, \left(\nabla_{\omega_N} \mathcal{V}_z\right)^\top\right]^\top$ $\in \mathbb{R}^{3N}$} and {\small $\Psi=\Bigg[\psi\left(\bar R^\top_1\nabla_{\bar R_1} \mathcal{V}_z\right)^\top,\psi\left(\bar R^\top_2\nabla_{\bar R_2} \mathcal{V}_z\right)^\top, \hdots, $ $\psi\left(\bar R^\top_M\nabla_{\bar R_M} \mathcal{V}_z\right)^\top\Bigg]^\top$ $\in \mathbb{R}^{3M}$}, where $\nabla_{\bar R_k} \mathcal{V}_z$ and $\nabla_{\omega_i} \mathcal{V}_z$ are the gradients of $ \mathcal{V}_z$ with respect to $\bar R_k$ and $\omega_i$, respectively, according to the \textit{Riemannian} metrics $\langle \eta_1, \eta_2 \rangle_{SO(3)} = \frac{1}{2} \operatorname{tr}(\eta_1^\top \eta_2)$ and $\langle y_1, y_2 \rangle_{\mathbb{R}^3} = y_1^\top y_2$ for every $\eta_1, \eta_2\in \mathfrak{so}(3)$ and $y_1, y_2 \in \mathbb{R}^3$. The critical points of the potential function $\mathcal{V}_z$ are given by the set $\{z \in \mathcal{S}_z : \Psi = 0 , \omega = 0\}$, which corresponds exactly to the set of equilibria $\Upsilon_z$ of the dynamics \eqref{R_bar_dynamics_k}-\eqref{w_dynamics_k} under the control torque \eqref{continuous_tau}. This completes the proof of item \eqref{itm_I1}.

To proceed with the proof of item \eqref{itm_I2} in Lemma \ref{lem_cp_pf}, we evaluate the \textit{Hessian} of \( \mathcal{V}_z(z) \), denoted as \( \text{\textit{Hess}}\,\mathcal{V}_z(z) \), to analyze the nature of the critical points of \( \mathcal{V}_z(z) \) that lie in the set \( \Upsilon_z \setminus \mathcal{A}_z \). Consider the two smooth curves defined in \eqref{curve_2}, where $\bar R_k=\bar R^*_k$ and $\omega_i=\omega_i^*$ with $(\bar{R}^*_1, \bar{R}^*_2, \ldots, \bar{R}^*_M, \omega^*_1, \omega^*_2, \ldots, \omega^*_N) \in \Upsilon_z \setminus \mathcal{A}_z$. The second derivative of $\mathcal{V}_z(\bar z)$ with respect to $t$ is given by:  
\begin{align}\label{hess}
    \frac{d^2}{dt^2}\mathcal{V}_z(\bar z)=&-k_R\sum_{k=1}^{M} \text{tr} \left(A\bar{R}^*_k\exp{\left(t [\zeta_k]^\times\right)}\left([\zeta_k]^\times\right)^2\right)\nonumber\\
    &-k_R\sum_{k=1}^{M} \text{tr}\left(A\bar{R}^*_k\exp{\left(t [\zeta_k]^\times\right)}[\dot \zeta_k]^\times\right)\nonumber\\
    &+\sum_{i=1}^N v_i^\top J_i v_i+\sum_{i=1}^N \dot{v}_i^\top J_i (\omega^*_i+v_i t).
\end{align}
Given that $\mathbb{P}_a(A\bar{R}^*_k) = 0$ and $\omega_i^* = 0$ for all $(\bar{R}^*_1, \bar{R}^*_2, \ldots, \bar{R}^*_M, \omega^*_1, \omega^*_2, \ldots, \omega^*_N) \in \Upsilon_z \setminus \mathcal{A}_z$, it follows that
{\small
\begin{align}\label{hess_2}
    \left.\frac{d^2}{dt^2}\mathcal{V}_z(\bar z)\right|_{t=0}=&-k_R\sum_{k=1}^{M} \text{tr} \left(A\bar{R}^*_k\left([\zeta_k]^\times\right)^2\right)+\sum_{i=1}^N v_i^\top J_i v_i.
\end{align}
}Using the identity $\left([\zeta]^\times\right)^2 = -\|\zeta\|^2 I_3 + \zeta \zeta^{\top}$ and the property $\text{tr}(\zeta_1 \zeta_2^{\top}) = \zeta_1^{\top} \zeta_2$ for all $\zeta, \zeta_1, \zeta_2 \in \mathbb{R}^3$, we further simplify equation \eqref{hess_2} as follows:
\begin{align}\label{hess_3}
    \left.\frac{d^2}{dt^2}\mathcal{V}_z(\bar z)\right|_{t=0}\!= k_R\sum_{k=1}^{M} \zeta^{\top}_k A^*_k \zeta_k+\sum_{i=1}^N v_i^\top J_i v_i,
\end{align}
where $A^*_k=\text{tr}(A\bar R^*_k)I_3 - A \bar R^*_k$. Letting $\zeta = [\zeta_1^\top, \zeta_2^\top, \hdots, \zeta_M^\top]^\top \in \mathbb{R}^{3M}$ and $v=[v_1^\top, v_2^\top, \ldots, v_N^\top]^\top \in \mathbb{R}^{3N}$, one has
\begin{align}\label{hess_4}
    \left.\frac{d^2}{dt^2}\mathcal{V}_z(\bar z)\right|_{t=0}\!&=k_R\,\zeta^\top \mathbf{A}^* \zeta+v^\top \mathbf{J} v\nonumber\\
    &=\begin{bmatrix}
        \zeta^\top&v^\top
    \end{bmatrix} \begin{pmatrix}
        k_R\,\mathbf{A}^*&0_{3M\times3N}\\
        0_{3N\times3M}&\mathbf{J}
    \end{pmatrix}\begin{bmatrix}
        \zeta \\ v
    \end{bmatrix}.
\end{align}
Recall that $\mathbf{A}^* = \text{diag}(A^*_1, A^*_2, \hdots, A^*_M)\in \mathbb{R}^{3M\times 3M}$. According to \cite{Mahony_book_OAMM}, it follows from \eqref{hess_4} that  
\begin{equation}\label{mtx_hess}
    \text{\textit{Hess}}\,\mathcal{V}_z(z)=\begin{pmatrix}
        k_R\,\mathbf{A}^*&0_{3M\times3N}\\
        0_{3N\times3M}&\mathbf{J}
    \end{pmatrix},
\end{equation}
for every $z \in \Upsilon_z \setminus \mathcal{A}_z$. The matrix \eqref{mtx_hess} represents the \textit{Hessian} of $\mathcal{V}_z$ evaluated at the critical points on $\Upsilon_z \setminus \mathcal{A}_z$. Note that $\text{\textit{Hess}}\,\mathcal{V}_z(z)$ is a block diagonal matrix. To determine whether the critical points on $\Upsilon_z \setminus \mathcal{A}_z$ are minima, maxima, or saddle points, it is essential to analyze the eigenvalues of the matrices $\mathbf{A}^*$ and $\mathbf{J}$. Given that the matrix $\mathbf{J}$ is positive definite, the focus should be on evaluating the eigenvalues of $\mathbf{A}^*$. Since $\mathbf{A}^* = \text{diag}(A^*_1, A^*_2, \hdots, A^*_M)$, we will explicitly determine the eigenvalues of the matrices \( A^*_k \) for every \( k \in \mathcal{M} \). According to the set $\Upsilon_z \setminus \mathcal{A}_z$, we have $A_m^* = \text{tr}(A) - A$ and $A_n^* \in \{{}^{1}A_n, {}^{2}A_n, {}^{3}A_n\}$ for every $m \in \mathcal{M}^I$ and $n \in \mathcal{M}^\pi$, where the matrix ${}^{\beta_n}A_n$ can be calculated as follows:
\begin{align}\label{A_equ}
    {}^{\beta_n}A_n = &\text{tr}\left(A \mathcal{R}(\pi, u_{\beta_n})\right)I_3 - A \mathcal{R}(\pi, u_{\beta_n}),
\end{align}
where $u_{\beta_n} \in \mathcal{E}(A)$. Using the fact that $\mathcal{R}(\pi, u_{\beta_n})=-I_3+2u_{\beta_n} u_{\beta_n}^\top$, it follows from \eqref{A_equ} that 
\[
{}^{\beta_n}A_n = -\text{tr}(A)I_3 + 2\lambda_{\beta_n}I_3 + A - 2\lambda_{\beta_n}u_{\beta_n}u_{\beta_n}^\top,
\]  
where $\lambda_{\beta_n}$ is the eigenvalue of $A$ corresponding to the eigenvector $u_{\beta_n}$ for each $\beta_n \in \{1, 2, 3\}$. Note that, for each \(n \in \mathcal{M}^\pi \), the matrix \( A_n^* \) can take one of three possible values (\ie, ${}^{1}A_n, {}^{2}A_n$ or ${}^{3}A_n$), depending on the choice of the eigenvector of the matrix \( A \). The set of eigenpairs of the matrix \( A_m^* \) is given by 
\[
\left\{(\lambda_2 + \lambda_3, u_1), (\lambda_1 + \lambda_3, u_2), (\lambda_1 + \lambda_2, u_3)\right\}, 
\]
for all \( m \in \mathcal{M}^I \). For the matrices \( {}^{1}A_n \), \( {}^{2}A_n \), and \( {}^{3}A_n \), the eigenpair sets  are found to be:
\begin{itemize}
    \item For \( {}^{1}A_n \): $\{(-\lambda_2 - \lambda_3, u_1), (\lambda_1 - \lambda_3, u_2), (\lambda_1 - \lambda_2, u_3)\},$
    \item For \( {}^{2}A_n \): $\{(\lambda_2 - \lambda_3, u_1), (-\lambda_1 - \lambda_3, u_2), (\lambda_2 - \lambda_1, u_3)\},$
    \item For \( {}^{3}A_n \): $\{(\lambda_3 - \lambda_2, u_1), (\lambda_3 - \lambda_1, u_2), (-\lambda_1 - \lambda_2, u_3)\},$
\end{itemize}
for all \( n \in \mathcal{M}^\pi \). It follows from the fact that \( \lambda_1 >\lambda_2 > \lambda_3 \), the eigenvalues of \( A^*_k \) must either be all negative or contain a mixture of positive and negative values. Combined with the fact that the matrix \( \mathbf{J} \) is positive definite, this indicates that the critical points of \( \mathcal{V}_z(z) \) within \( \Upsilon_z \setminus \mathcal{A}_z \) are saddle points of \( \mathcal{V}_z(z) \). This completes the proof of Lemma \ref{lem_cp_pf}.

\section{Proof of Theorem \ref{theorem1}} \label{app_2}
In view of \eqref{w_dynamics_k} and \eqref{R_obs_f}, one can verify that 
    \begin{equation}\label{w_total_dyn}
        \mathbf{J} \dot \omega = -k^R \bar{H}~\Psi_\nabla^{\bar R}-k_\omega \omega - \bar k_\omega (\mathcal{L}\otimes I_3) \omega,
    \end{equation}
    where {\small $\Psi_\nabla^{\bar R} :=\Bigg[\psi\left(\bar R^\top_1\nabla_{\bar R_1}\bar U\right)^\top ,\psi\left(\bar R^\top_2\nabla_{\bar R_2}\bar U\right)^\top , \hdots,$ $ \psi\left(\bar R^\top_M\nabla_{\bar R_M}\bar U\right)^\top \Bigg]^\top  \in \mathbb{R}^{3M}$}, $\mathbf{J} :=\text{diag}(J_1, J_2, \hdots, J_N)\in \mathbb{R}^{3N\times 3N}$, $\mathcal{L}:= H H^\top  \in \mathbb{R}^{N\times N}$ is the Laplacian matrix corresponding to the graph $\mathcal{G}$ and the matrix $\bar H$ is given in (\ref{H_bar}). Consider the following Lyapunov function candidate:
    \begin{align}\label{potential_fct_V}
        \mathcal{V}(\bar x)=&k_R \bar U(x)+\omega^\top  \mathbf{J}  \omega.
    \end{align}
    Note that $\mathcal{V}$ is positive definite on $\bar{\mathcal{S}}$ with respect to $\bar{\mathcal{A}}$. The time-derivative of $\mathcal{V}$, along the trajectories generated by the flows of the hybrid closed-loop dynamics (\ref{hybrid_sys}), is given by
    \begin{align}\label{v_dot_}
        \dot{\mathcal{V}}(\bar x)= k_R \dot{\bar U}(x)+2 \omega^\top  \mathbf{J}  \dot{\omega}.
    \end{align}
    The time-derivative of the first term of \eqref{v_dot_} can be calculated as follows
    {\small
     \begin{align}\label{U_grd}
        \dot{\bar U}(x) =& \sum_{k=1}^{M}\langle \nabla_{\bar R_k}\bar U, \bar R_k [\bar \omega_k]^\times\rangle_{\bar R_k}+\sum_{k=1}^{M}\langle \langle \nabla_{\xi_k}\bar U, \dot \xi_k \rangle \rangle \nonumber\\
        =& \sum_{k=1}^{M}\langle \langle \bar R_k^\top  \nabla_{\bar R_k}\bar U, [\bar \omega_k]^\times\rangle \rangle+\sum_{k=1}^{M}\langle \langle \nabla_{\xi_k}\bar U, \dot \xi_k \rangle \rangle \nonumber\\
        =& 2\sum_{k=1}^{M} \bar \omega_k^\top  \psi\left(\bar R^\top_k\nabla_{\bar R_k}\bar U\right)+\sum_{k=1}^{M}\dot \xi_k\nabla_{\xi_k}\bar U\\
        =& 2~\bar \omega^\top  \Psi_\nabla^{\bar R}+\dot \xi^\top  \Psi_\nabla^\xi= 2~\omega^\top  \bar H \Psi_\nabla^{\bar R}-k_\xi ||\Psi_\nabla^\xi||^2,\label{U_grd}
    \end{align}
    }where {\small$\Psi_\nabla^\xi :=\left[\nabla_{\xi_1}\bar U,\nabla_{\xi_2}\bar U, \hdots, \nabla_{\xi_M}\bar U\right]^\top  \in \mathbb{R}^M$}. To derive the above equations, the following identities have been used: $\langle \eta_1, \eta_2 \rangle_R=\langle\langle R^\top \eta_1, R^\top \eta_2 \rangle\rangle$, $\text{tr}\left(B[x]^\times\right)=\text{tr}\left(\mathbb{P}_a(B)[x]^\times\right)$ and $\text{tr}\left([x]^\times [y]^\times\right)=-2x^\top y$, $\forall x, y\in \mathbb{R}^3$, $\forall B \in \mathbb{R}^{3\times3}$, $\forall \eta_1, \eta_2 \in \mathfrak{so}(3)$ and $\forall R \in SO(3)$. Furthermore, from \eqref{w_dynamics_k}, \eqref{R_obs_f} and \eqref{U_grd}, one obtains the following time-derivative of $\mathcal{V}(\bar x)$ 
    {\small
    \begin{align}\label{v_dot}
        \dot{\mathcal{V}}(\bar x)=&-k_R k_\xi ||\Psi_\nabla^\xi||^2- 2 k_\omega ||\omega||^2 -2 \bar k_\omega ||(H^\top \otimes I_3) \omega||^2,
    \end{align}
    }which implies that $\mathcal{V}$ is non-increasing along the flows of (\ref{hybrid_sys}). Moreover, in view of (\ref{hybrid_sys}) and (\ref{potential_fct_V}), one has
    {\small
    \begin{align}\label{v_+}
        \mathcal{V}(\bar x)&-\mathcal{V}(\bar x^+)=k_R \left( \bar U(x)-\bar U(x^+)\right)\nonumber\\
        &=k_R\sum_{k=1}^{M}\left(U(\bar{R}_k, \xi_k)-U(\bar{R}_k^+, \xi_k^+)\right)
        \geq& k_R \delta_{\bar R},
    \end{align}
    }which indicates that $\mathcal{V}(\bar x)$ is strictly decreasing over the jumps of (\ref{hybrid_sys}). In view of \eqref{v_dot} and \eqref{v_+}, it follows that set $\bar{\mathcal{A}}$ is stable \cite[Theorem 23]{Goebel_ieee_magazine}, and hence all maximal solutions of \eqref{hybrid_sys} are bounded. From (\ref{v_dot}) and (\ref{v_+}), one can also verify that $\mathcal{V}(\bar x(t,j))\leq \mathcal{V}(\bar x(t_j,j))$ and $\mathcal{V}(\bar x(t_j,j))\leq \mathcal{V}(\bar x(t_j,j-1))-k_R \delta_{\bar R}$, $\forall (t,j), (t_j,j), (t_j,j-1) \in \text{dom}~ \bar x$, with $(t,j)\geq (t_j,j) \geq (t_j,j-1)$. Thus, one has $0\leq \mathcal{V}(\bar x(t,j))\leq \mathcal{V}(\bar x(0,0))-jk_R$, $\forall (t,j)\in \text{dom}~ \bar x$, which leads to $j\leq \lceil \frac{\mathcal{V}(\bar x(0,0))}{k_R\delta_{\bar R}}\rceil$, where $\lceil . \rceil$ denotes the ceiling function. The last inequality implies that the number of jumps is finite and depends on the initial conditions.\\
   Next, using the invariance principle for hybrid systems \cite[Section 8.2]{goebel2012hybrid}, we will demonstrate the global attractivity of $\bar{\mathcal{A}}$. Defining the following functions:
    
    {\small
    \begin{align}
         u_{\bar{\mathcal{F}}}(\bar x) &:=\begin{cases}
                        -k_R k_\xi ||\Psi_\nabla^\xi||^2- 2 k_\omega ||\omega||^2 -2 \bar k_\omega ||(H^\top \otimes I_3) \omega||^2\\
                        ~~~~~~~~~~~~~~~~~~~~~~~ \text{if}~\bar x \in \bar{\mathcal{F}},\\
                         -\infty~~~~~~~~~~~~ \text{otherwise},
                    \end{cases}\label{uf}\\
         u_{\bar{\mathcal{J}}}(\bar x) &:=\begin{cases}
                        &-k_R \delta_{\bar R}~~~~~~~~~~~~ \text{if}~\bar x \in \bar{\mathcal{J}},\\
                        & -\infty~~~~~~~~~~ \text{otherwise}.
                    \end{cases}\label{uj}
    \end{align}}One verifies that the growth of $\mathcal{V}$ is upper bounded during the flows by $u_{\bar{\mathcal{F}}}(\bar x) \leq 0$ and during the jumps by $u_{\bar{\mathcal{J}}}(\bar x) \leq 0$ for every $\bar x \in \bar{\mathcal{S}}$. Consequently, according to \cite[Corollary 8.4]{goebel2012hybrid}, every maximal solution of the hybrid system \eqref{hybrid_sys} converges to the following largest weakly\footnote{The reader is referred to \cite{goebel2012hybrid} for the definition of \textit{weakly invariant} sets in the hybrid systems context.} invariant subset: $$\mathcal{V}^{-1}(r)\cap \bar{\mathcal{S}} \cap \left[\overline{u_{\bar{\mathcal{F}}}^{-1}(0)} \cup \left(u_{\bar{\mathcal{J}}}^{-1}(0)\cap \bar G \left(u_{\bar{\mathcal{J}}}^{-1}(0)\right)\right)\right],$$ for some $r \in \mathbb{R}$, where $\overline{u_{\bar{\mathcal{F}}}^{-1}(0)}$ denotes the closure of the set $u_{\bar{\mathcal{F}}}^{-1}(0)$. For $k_\omega >0$, one can verify that 
\begin{align}
    u_{\bar{\mathcal{F}}}^{-1}(0)=\{\bar x \in \bar{\mathcal{F}}:~\Psi_\nabla^\xi=0,~ \omega=0\},\quad
    u_{\mathcal{J}}^{-1}(0)= \emptyset. \nonumber
\end{align}
Since for every $\bar x \in u_{\bar{\mathcal{F}}}^{-1}(0)$, one has $\omega=0$, which implies $\dot \omega =0$, it follows from \eqref{w_total_dyn} that $\bar H \Psi_\nabla^{\bar R} =0$, which further implies, as per \cite[Lemma 2]{Mouaad_ACC23}, that $\Psi_\nabla^{\bar R}=0$. Therefore, one can verify that any solution to the hybrid closed-loop system \eqref{hybrid_sys} converges to the largest invariant set contained in 
\begin{align}
    u_{\bar{\mathcal{F}}}^{-1}(0)&=\{\bar x \in \bar{\mathcal{F}}:~x \in \mathcal{F}\cap \bar \Upsilon, ~\omega=0\}. \nonumber
\end{align}On the other hand, given $x \in \mathcal{A}$, one has, for all $k \in \mathcal{M}$, $U(\bar{R}_k, \xi_k)-\underset{\bar{\xi}_k\in \Xi}{\text{min}} U(\bar{R}_k,\bar{\xi}_k)=-\underset{\bar{\xi}_k\in \Xi}{\text{min}} U(\bar{R}_k,\bar{\xi}_k)\leq 0$. Therefore, from (\ref{network_f_j_set}), and according to Condition \ref{hybrid_ass}, one can verify that $\mathcal{A} \subset \mathcal{F}\cap \bar \Upsilon$ and $\mathcal{F}\cap (\bar \Upsilon \setminus \mathcal{A})=\varnothing$. In addition, applying some set-theoretic arguments, one has $\mathcal{F} \cap \bar \Upsilon \subset (\mathcal{F} \cap (\bar \Upsilon\setminus \mathcal{A})) \cup (\mathcal{F} \cap \mathcal{A})= \varnothing \cup \mathcal{A}$. It follows from $\mathcal{A} \subset \mathcal{F} \cap \bar \Upsilon$ and $\mathcal{F} \cap \bar \Upsilon \subset \mathcal{A}$ that $\mathcal{F} \cap \bar \Upsilon= \mathcal{A}$. This implies that $u_{\bar{\mathcal{F}}}^{-1}(0)=\bar{\mathcal{A}}$.  Hence, every maximal solution of the hybrid system \eqref{hybrid_sys} converges to the largest weakly invariant subset $\mathcal{V}^{-1}(0)\cap \bar{\mathcal{A}} = \bar{\mathcal{A}}$. Since every maximal solution of the hybrid closed-loop system (\ref{hybrid_sys}) is bounded, $\bar G(\bar x) \in \bar{\mathcal{F}} \cup \bar{\mathcal{J}}$ for every $\bar x \in \bar{\mathcal{J}}$, and $\bar F(\bar x)\subset T_{\bar{\mathcal{F}}}(\bar x)$, for every $\bar x \in \bar{\mathcal{F}}\setminus \bar{\mathcal{J}}$, where $T_{\bar{\mathcal{F}}}(\bar x)$ denotes the tangent cone to $\bar{\mathcal{F}}$ at the point $\bar{x}$, according to \cite[Proposition 6.10]{goebel2012hybrid}, one can conclude that every maximal solution of the hybrid closed-loop system (\ref{hybrid_sys}) is complete. This, together with Lemma \ref{lem_hbc}, allows us to conclude that the set $\bar{\mathcal{A}}$ is globally asymptotically stable for the hybrid closed-loop system (\ref{hybrid_sys}). This completes the proof.

\section{Proof of Theorem \ref{theorem2}} \label{app_3}
 Consider the following Lyapunov function candidate: 
 {\small
\begin{align*}
       \hat{\mathcal{V}}(\hat x)=& k_R \sum_{i=1}^M U(\bar R_k, \xi_k)+k_{\tilde Q} \sum_{i=1}^N U(\tilde Q_i, \zeta_i) + \sum_{i=1}^N \omega_i^\top  J_i \omega_i.
   \end{align*}
}
   One can verify that the above Lyapunov function candidate is positive definite on $\hat{\mathcal{S}}$ with respect to $\hat{\mathcal{A}}$, and its time-derivative, along the trajectories generated by the flows of the hybrid closed-loop dynamics \eqref{hybrid_sys_w_v}, is given by
  {\small
   \begin{align}
       \dot{\hat{\mathcal{V}}}(\hat x)=-k_R k_\xi ||\Psi_\nabla^\xi||^2-2 k_{\tilde Q} k_Q ||\Psi_\nabla^{\tilde Q}||^2-k_{\tilde Q} k_\zeta||\Psi_\nabla^\zeta||^2,
   \end{align}
}
   where {\small $\Psi_\nabla^{\tilde Q} :=\Bigg[\psi\left(\tilde Q^\top_1\nabla_{\tilde Q_1}U(\tilde Q_1, \zeta_1)\right)^\top \hspace{-0.2cm},\psi\left(\tilde Q^\top_2\nabla_{\tilde Q_2}U(\tilde Q_2, \zeta_2)\right)^\top $ $, \hdots, \psi\left(\tilde Q^\top_N\nabla_{\tilde Q_N}U(\tilde Q_N, \zeta_N)\right)^\top \Bigg]^\top  \in \mathbb{R}^{3N}$, $\Psi_\nabla^\zeta :=\bigg[\nabla_{\zeta_1}U(\tilde Q_1, \zeta_1),\nabla_{\zeta_2}U(\tilde Q_2, \zeta_2), \hdots, \nabla_{\zeta_N}U(\tilde Q_N, \zeta_N)\bigg]^\top  $ $ \in \mathbb{R}^N$} and $\zeta :=[\zeta_1, \zeta_2, \hdots, \zeta_N]^\top \in \mathbb{R}^N$. This implies that $\hat{\mathcal{V}}$ is non-increasing along the flow of \eqref{hybrid_sys_w_v}. Furthermore, one has 
   \begin{align}
       \hat{\mathcal{V}}(\hat x)-&\hat{\mathcal{V}}(\hat x^+)=k_R \sum_{k=1}^{M}\left(U(\bar R_k, \xi_k)-U(\bar R_k^+, \xi_k^+)\right)\nonumber\\
       &+k_{\tilde Q} \sum_{i=1}^{N}\left(U(\tilde Q_i, \zeta_i)-U(\tilde Q_i^+, \zeta_i^+)\right)
       \geq \underline{k}~\underline \delta
   \end{align}
   where $\underline k :=\text{min}\{k_R, k_{\tilde Q}\}$ and $\underline \delta :=\text{min}\{\delta_{\bar R}, \delta_{\tilde Q}\}$. Following the same steps as in the proof of Theorem \ref{theorem1}, it can be shown that the set $\hat{\mathcal{A}}$ is stable, every maximal solution of the hybrid closed-loop dynamics \eqref{hybrid_sys_w_v} is complete, and the number of jumps is finite. Furthermore, consider the following two functions:
   \begin{align}
     u_{\hat{\mathcal{F}}}(\hat x) &:=\begin{cases}
                     -k_R k_\xi ||\Psi_\nabla^\xi||^2-2 k_{\tilde Q} k_Q ||\Psi_\nabla^{\tilde Q}||^2-k_{\tilde Q} k_\zeta||\Psi_\nabla^\zeta||^2\nonumber\\
                     ~~~~~~~~~~~~~~~~~~~~\text{if}~\hat x \in \hat{\mathcal{F}},\\
                     -\infty~~~~~~~~~~~~ \text{otherwise},
                \end{cases}\\
     u_{\hat{\mathcal{J}}}(\hat x) &:=\begin{cases}
                    &-\underline{k}~\underline \delta~~~~~~~~~~~~ \text{if}~\hat x \in \hat{\mathcal{J}},\\
                    & -\infty~~~~~~~~~~ \text{otherwise}.
                \end{cases}
\end{align}
It follows from the invariance principle for hybrid systems, given in \cite[Section 8.2]{goebel2012hybrid}, that every maximal solution of the hybrid system \eqref{hybrid_sys_w_v} converges to the following largest weakly invariant subset: $$\hat{\mathcal{V}}^{-1}(r)\cap \hat{\mathcal{S}} \cap \left[\overline{u_{\hat{\mathcal{F}}}^{-1}(0)} \cup \left(u_{\hat{\mathcal{J}}}^{-1}(0)\cap G \left(u_{\hat{\mathcal{J}}}^{-1}(0)\right)\right)\right],$$ for some $r \in \mathbb{R}$, where $u_{\hat{\mathcal{J}}}^{-1}(0)=\emptyset$ and $u_{\hat{\mathcal{F}}}^{-1}(0)=\{\hat x \in \hat{\mathcal{F}}:~\Psi_\nabla^\xi=0,~\Psi_\nabla^{\tilde Q}=0,~\Psi_\nabla^\zeta=0\}$. Note that, for every $\hat x \in u_{\hat{\mathcal{F}}}^{-1}(0)$, one has $\Psi_\nabla^\xi=0$ and, for $i \in \mathcal{V}$, $(\tilde Q_i, \zeta_i) \in \mathcal{F}_i^{\tilde Q} \cap \Upsilon$. According to \cite{miaomiao_TAC2022}, along with Condition \ref{hybrid_ass_2}, it can be shown that $\mathcal{F}_i^{\tilde Q} \cap \Upsilon = \{(I_3, 0)\}$. Moreover, from the fact that $\dot{\tilde Q}_i=0$ (since $\tilde Q_i = I_3$), one has $\omega_i = k_Q \psi\left(\tilde Q^\top_i \nabla_{\tilde Q_i}U(\tilde Q_i, \zeta_i)\right)=0$. This implies that $k_R\bar H \Psi_\nabla^{\bar R}=0$. This fact together with $\Psi_\nabla^\xi=0$ and considering the last part of the proof of Theorem \ref{theorem1}, one has $\bar x \in \bar{\mathcal{A}}$. Finally, one concludes that $u_{\hat{\mathcal{F}}}^{-1}(0)=\hat{\mathcal{A}}$ and every maximal solution of the hybrid system \eqref{hybrid_sys_w_v} converges to the largest weakly invariant subset $\hat{\mathcal{V}}^{-1}(0)\cap \hat{\mathcal{A}} = \hat{\mathcal{A}}$. Combining this with the fact that every maximal solution of the hybrid closed-loop system (\ref{hybrid_sys_w_v}) is complete and (\ref{hybrid_sys_w_v}) satisfies the basic hybrid conditions, implies that the set $\hat{\mathcal{A}}$ is globally asymptotically stable for the hybrid closed-loop system (\ref{hybrid_sys_w_v}). This completes the proof.

\bibliographystyle{IEEEtran}
\bibliography{References}
\end{document}